\documentclass{IEEEtran}
\renewcommand{\baselinestretch}{0.99}
\usepackage{setspace}
\usepackage{balance}
\usepackage{graphics}
\usepackage{graphicx}

\usepackage[caption=false,font=footnotesize]{subfig}
\usepackage[bookmarks=false]{hyperref}
\usepackage{epsfig}
\usepackage{times} 
\usepackage{amsthm}
\usepackage{arydshln}
\usepackage{amsmath} 
\usepackage{amssymb}
\usepackage{amsfonts}
\usepackage{verbatim}
\usepackage{mathtools}
\usepackage{diagbox}
\usepackage{cite}
\usepackage{multicol}
\usepackage{multirow}
\usepackage{url}
\usepackage[caption=false,font=footnotesize]{subfig}
\usepackage{mathrsfs}
\usepackage{comment}
\usepackage{tikz-cd}
\usepackage{tikz}
\usepackage{circuitikz}
\usepackage{tikz-3dplot}
\tdplotsetmaincoords{60}{6}
\usetikzlibrary{3d}
\usetikzlibrary{arrows.meta, positioning, calc}
\usepackage{pgfplots}
\pgfplotsset{compat=newest}
\usepackage{url}
\usepackage{enumerate}
\usepackage[normalem]{ulem}
\usepackage{xcolor}
\colorlet{ab}{teal}
\usepackage{nicematrix}
\allowdisplaybreaks
\usepackage[ruled,longend]{algorithm2e}

\newcommand{\pr}[1]{\left( #1\right)}
\newcommand{\ang}[1]{\left\langle #1\right\rangle}
\newcommand{\matrices}[1]{\begin{bmatrix} #1\end{bmatrix}}

\newtheorem{assumption}{Assumption}

\newtheorem{remark}{Remark}
\newtheorem{theorem}{Theorem}
\newtheorem{lemma}{Lemma}
\newtheorem{corollary}{Corollary}

\newtheorem{definition}{Definition}

\DeclareMathOperator{\im}{Im}
\DeclareMathOperator{\Ker}{Ker}

\newcommand{\R}{\mathbb{R}}
\newcommand{\C}{\mathbb{C}}

\newcommand{\rank}{\mathrm{rank}\, }
\newcommand{\kernel}{\mathrm{Ker}\, }
\newcommand{\image}{\mathrm{Im}\, }
\newcommand{\sign}{\mathrm{sign}}
\newcommand{\B}{\mathscr{B}}
\newcommand{\X}{\mathscr{X}}
\newcommand{\Y}{\mathscr{Y}}
\newcommand{\U}{\mathscr{U}}
\newcommand{\V}{\mathscr{V}}

\newcommand{\Ss}{\mathscr{S}}

\newcommand{\W}{\mathscr{W}}
\newcommand{\M}{\mathscr{M}}
\newcommand{\K}{\mathscr{K}}

\newcommand{\Lap}{\mathcal{L}}
\newcommand{\diag}{\mathrm{diag}}
\newcommand{\col}{\mathrm{col}}

\newcommand{\myquad}[1][1]{\hspace*{#1em}\ignorespaces}

\def\BibTeX{{\rm B\kern-.05em{\sc i\kern-.025em b}\kern-.08em
    T\kern-.1667em\lower.7ex\hbox{E}\kern-.125emX}}
\begin{document}
\title{Distributed Unknown Input Observer Design: \\ A Geometric Approach}
\author{
Ruixuan Zhao, Guitao Yang, Thomas Parisini, \IEEEmembership{Fellow, IEEE} and Boli Chen, \IEEEmembership{Senior Member, IEEE
}
\thanks{This work has been partially supported by European Union's Horizon 2020 research and innovation programme under grant agreement no. 739551 (KIOS CoE).}%
	\thanks{R. Zhao and B. Chen are with the Department of Electronic and Electrical Engineering, University College London, London, UK \tt\small(ruixuan.zhao.22@ucl.ac.uk; boli.chen@ucl.ac.uk).}%
    \thanks{G. Yang is with the Wolfson School of Mechanical, Electrical and Manufacturing Engineering, Loughborough University, Loughborough, UK \tt\small(g.yang@lboro.ac.uk).}
	\thanks{T. Parisini is with the Department of Electrical and Electronic Engineering, Imperial College London, London. T. Parisini is also with the Department of Electronic Systems, Aalborg University, Denmark, and with the Department of Engineering and Architecture, University of Trieste, Italy \tt\small( t.parisini@imperial.ac.uk).}
 }

\maketitle

\begin{abstract}
We present a geometric approach to designing distributed unknown input observers (DUIOs) for linear time-invariant systems, where measurements are distributed across nodes and each node is influenced by \emph{unknown inputs} through distinct channels.
The proposed distributed estimation scheme consists of a network of observers, each tasked with reconstructing the entire system state despite having access only to local input-output signals that are individually insufficient for full state observation. Unlike existing methods that impose stringent rank conditions on the input and output matrices at each node, our approach leverages the $(C,A)$-invariant (conditioned invariant) subspace at each node from a geometric perspective. This enables the design of DUIOs in both continuous- and discrete-time settings under relaxed conditions, for which we establish sufficiency and necessity.
The effectiveness of our methodology is demonstrated through extensive simulations, including a practical case study on a power grid system.
\end{abstract}


\section{Introduction}\label{sec:Intro}
As Cyber-Physical Systems (CPS) increase in scale and complexity -- particularly in critical infrastructures -- the need for state observers to monitor system operations becomes increasingly important for ensuring reliability and the early detection of system faults or malicious cyber-attacks \cite{saidi2018intergrated}. 
Typically, these large-scale systems are monitored by a network of sensors, as a single sensor (node) is often inadequate for observing the state of the entire system \cite{xie2012fully,ahmed2016distributed}. The proliferation of sensing equipment, coupled with advancements in communication and computation technologies, enables the deployment of distributed algorithms directly on intelligent sensors. Consequently, it is highly advantageous to develop distributed schemes that leverage sensors with computational capabilities available on board to collaboratively reconstruct the entire system state. However, in large-scale CPS, individual nodes often do not have access to all inputs that affect system dynamics, making it crucial to propose a distributed state estimation method that accounts for unknown inputs.

We consider the problem of estimating the state vector of a linear time-invariant (LTI) plant in both continuous-time
\begin{equation}\label{eq:con_sys}
    \dot x (t) = Ax(t) + Bu(t) \, ,
\end{equation}
and discrete-time form
\begin{equation}\label{eq:dis_sys}
     x (t+1) = Ax(t) + Bu(t) \, ,
\end{equation}
where ${x \in \mathbb{R}^n}$ is the state vector, ${u\in \mathbb{R}^m}$ is the control input. ${A\in \mathbb{R}^{n\times n}}$ and ${B\in \mathbb{R}^{n\times m}}$ are known state and input matrices, respectively.
The system state is measured by a bank of sensors
\begin{equation}\label{eq:output}
    y_i(t) = C_i x(t) \, ,
\end{equation}
where $y_i \in \mathbb{R}^{p_i}$ is the output measurement at node~$i$, $C_i \in \mathbb{R}^{p_i\times n}$ and $i\in\mathbf{N}=\{1,2,\cdots,N\}$ is the nonempty set of sensor nodes.

At node $i$, the system input $u$ can be partitioned into two parts: $u_i\in \mathbb{R}^{l_i}$ and $\bar u_i \in \mathbb{R}^{m - l_i}$, where $u_i$ collects all known input signals from node $i$, while $\bar u_i$ 
consists of the remaining components of \(u\), which are unknown to node \(i\).
Accordingly, by proper partitioning of $B$ and $u$, 
we have
\begin{equation}\label{eq:inputdecomp}
    Bu(t) = B_i u_i(t) +\bar{B}_i \bar{u}_i(t),
\end{equation}
with $B_i \in \mathbb{R}^{n\times l_i}$, $\bar{B}_i \in \mathbb{R}^{n\times (m - l_i)}$, $u_i \in \mathbb{R}^{l_i}$, and $\bar u_i \in \mathbb{R}^{m - l_i}$, where $m-l_i \leq p_i \leq n$.
Indeed, matrix $B_i$ represents the actuation channel of inputs that are known to node $i$ and hence $u_i$ can be used locally for state estimation. 
Conversely, $\bar{u}_i$ represents the components of \(u\) that are complementary to $u_i$.

In this context, systems \eqref{eq:con_sys} and \eqref{eq:dis_sys} can be reformulated for each node $i \in \mathbf{N}$ as
\begin{subequations}\label{eq:systemDecomposition}
    \begin{align}
        &\dot x(t) = Ax(t) + B_iu_i(t) + \bar B_i \bar{u}_i(t),\label{eq:systemDecomposition_ct}\\
        &x(t+1) = Ax(t) + B_iu_i(t) + \bar B_i \bar{u}_i(t).\label{eq:systemDecomposition_dt}
    \end{align}
\end{subequations}
\emph{A distributed unknown input observer (DUIO) consists of a set of observers $\mathcal O = \{ \mathcal O_i \}_{i \in \mathbf N}$, each associated with a sensor node, that are capable of collectively reconstructing the state vector $x$ in the presence of unknown inputs $\bar{u}_i$ at each node.
Each local observer $\mathcal O_i$ can only access its local input $u_i$ and output $y_i$, and can exchange its local estimate with neighboring nodes following an undirected communication graph $\mathcal{G}$ (see Fig.~\ref{fig:network} for an illustrative example).}

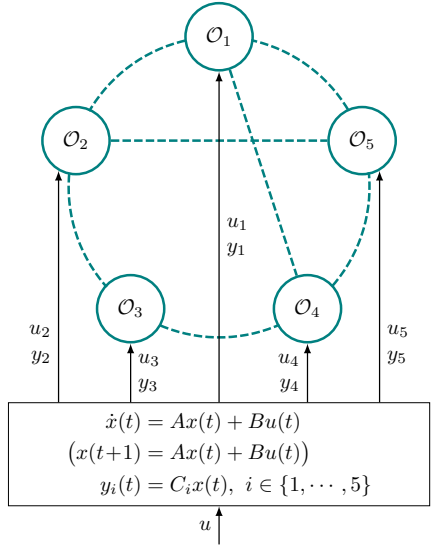
\begin{figure}[htp]
    \centering
    \scalebox{0.8}{\begin{tikzpicture}
\def\off{18}
\def\N{5}
\def\R{2.5}
\pgfmathparse{360/\N}
\edef\step{\pgfmathresult}

\colorlet{net}{teal}
\tikzset{comm/.style = {color=net, very thick, dash pattern=on 4pt off 1.5pt}}
\coordinate (O) at (0,0);
\foreach \x in {1,...,\N} {%
    \node [circle, 
        draw, 
        color=net, 
        fill=white, 
        text=black, 
        very thick,
        inner sep=6pt] (N\x) at (\x*\step+\off:\R) {$\mathcal O_{\x}$};
}
\pgfmathanglebetweenpoints{
    \pgfpointanchor{N1}{west}}{
    \pgfpointanchor{O}{center}}
\edef\anglestart{\pgfmathresult}
\pgfmathanglebetweenpoints{
    \pgfpointanchor{N2}{north}}{
    \pgfpointanchor{O}{center}}
\edef\angleend{\pgfmathresult-0.7}
\foreach \x in {1,...,\N} {
    \edef\rotation{\step*(\x-1)}
    \draw[comm,rotate=\rotation] (\anglestart-180:\R) arc[radius = \R, start angle = \anglestart-180, end angle = \angleend-180];
}
\draw [comm] (N2) -- (N5);
\draw [comm] (N1) -- (N4);


\node [draw,
        rectangle,
        below=2.2*\R of N1,
        minimum width=7cm,
        minimum height=1.2cm] (sys) {$\begin{aligned}
            \dot x(t) &= Ax(t) + Bu(t) \\
            \big(x(t\!+\!1) &= Ax(t) + Bu(t)\big)\\
            y_i(t) &= C_ix(t),\ i\in\{1,\cdots,5\}
        \end{aligned}$}; 

\draw [-latex, semithick] (sys.north -| N1) -- (N1) node[midway, right, align=left] {$u_1$ \\ $y_1$};
\draw [-latex, semithick] (sys.north -| N2.240) -- (N2.240) node[near start, left, align=right] {$u_2$ \\ $y_2$};
\draw [-latex, semithick] (sys.north -| N3) -- (N3) node[midway, right, align=left] {$u_3$ \\ $y_3$};
\draw [-latex, semithick] (sys.north -| N4) -- (N4) node[midway, left, align=right] {$u_4$ \\ $y_4$};
\draw [-latex, semithick] (sys.north -| N5.300) -- (N5.300) node[near start, right, align=left] {$u_5$ \\ $y_5$};

\coordinate (bot) at ($(sys)+(0,-1.5)$);
\draw [-latex, semithick] ([shift={(0,0)}]bot) -- ([shift={(0,0)}]bot |- sys.south) node[midway, left] {$u$};
\end{tikzpicture}}\\[-2.5ex]
    \caption{Example of a distributed unknown input observer consisting of 5 sensor nodes and thus a group of 5 local observers $\mathcal O_1,\dots \mathcal O_5$. The communication graph is represented by the dashed lines.
    }
    \label{fig:network}
\end{figure}
\subsection{Brief Survey of Related Works}
Early research on distributed state estimation primarily focused on extending traditional centralized state estimation methods for linear systems, adapting them to networked settings by enabling information exchange over communication links. For example, in \cite{kamgarpour2008convergence, olfati2007distributed, olfati2009kalman}, the classical Kalman filtering framework is extended to distributed estimation by incorporating consensus among neighboring estimators, thereby addressing the locally unobservable components at each node. Similar consensus-based strategies have been explored in \cite{han2018simple, kim2019completely}, where distributed Luenberger-like observers were developed for continuous-time LTI systems. Building on this idea, \cite{kim2019completely} and \cite{yang2023plug} introduced adaptive consensus gains, enabling fully decentralized designs of distributed observers that support plug-and-play functionality, i.e., the seamless addition or removal of sensor nodes without requiring observer redesign. In contrast, by transforming the continuous-time LTI system into its Jordan canonical form, a recent work \cite{zhang2022decentralized} introduces an alternative decentralized design for distributed observers. An augmented state observer mechanism is proposed in \cite{park2016design} for discrete-time systems and in \cite{wang2017distributed} for continuous-time systems. In particular, \cite{park2016design} establishes the necessary and sufficient conditions for the estimation parameter selection in the distributed observer design process. Moreover, \cite{mitra2018distributed} proposes a method for discrete-time distributed observer design based on successive observability decomposition, while more recently \cite{wang2024split} introduces a split-spectrum approach applicable to both continuous- and discrete-time LTI systems.
The aforementioned methods ensure asymptotic convergence to the true system state; however, in certain applications, finite-time or fixed-time convergence may be required. This has been investigated, for example, in \cite{silm2018note} and \cite{ge2022fixed}. In addition, \cite{mitra2021distributed, wang2022hybrid, yang2023state, liu2023distributed, liu2024distributed} develop distributed observer designs for systems with switching communication topologies, which are particularly relevant for mobile sensor networks. While most existing works concentrate on linear systems, a few recent studies, such as 
\cite{wu2021design,yang2022sensor,cao2025design}, have extended distributed observer design to nonlinear systems under specific continuity assumptions.

It is worth noting that all the aforementioned works focus on scenarios where either no input signals are present or the input information is fully available to all nodes. This assumption can be rather restrictive for realistic large-scale distributed applications. In many CPS of practical interest, individual nodes typically have access only to partial input information (e.g., local inputs), which motivates research on the DUIO problem. Building on the unknown input observer (UIO) framework developed for centralized systems (see, e.g., \cite{chen1996design,darouach1994full,kudva2003observers,zhu2023asymptotic,ichalal2019unknown,disaro2024equivalence}), recent efforts on DUIO design have emerged. In particular, DUIOs have been studied in the context of continuous-time systems \cite{yang2022state,cao2023distributed,cao2023distributedauto,zhu2024distributed} and, more recently, for discrete-time systems in \cite{disaro2024distributed}. 
A common underlying assumption across these works is that:
\begin{equation}
    \rank (C_i\bar{B}_i)=\rank (\bar{B}_i) \, ,\label{eq:rank_con} 
\end{equation}
or, equivalently, 
\begin{equation*}
    \im(\bar{B}_i)\cap \kernel(C_i) = 0 \, ,
\end{equation*}
with $\kernel (\cdot)$ denoting the kernel (null space) and $\image (\cdot)$ denoting the image (column space) of a linear map.
Recall the geometric condition for left-invertibility \cite[Proposition~9]{massoumnia1986geometric} for system $(C_i,A,\bar{B}_i)$
\begin{equation*}
    \im(\bar{B}_i)\cap \V^*_i(\kernel{C_i}) = 0,
\end{equation*}
where $\V^*_i(\kernel{C_i})$ is the supremal controllability subspace contained in $\kernel{C_i}$. 
It is immediate to show that systems $(C_i,A,\bar{B}_i)$ satisfying the rank condition \eqref{eq:rank_con} are indeed \textit{left-invertible} with a monic $\bar B_i$, which, in turn, implies that the unknown input $\bar{u}_i$ can be uniquely determined by the local output $y_i$. 
However, for a large-scale CPS, the requirement that each local node $(C_i,A,\bar{B}_i)$ is left-invertible is extremely restrictive, as shown, for example, by the practical example in Section~\ref{exam:DGU}. 
Recent work \cite{cao2023bounded} has sought to relax the rank condition in continuous-time DUIO design. However, the conditions under which a DUIO can be successfully designed remain insufficiently characterized.

\subsection{Main Contributions and Paper Organization}
Motivated by the limitations of existing DUIO designs outlined in the previous section, this article proposes a novel DUIO framework that achieves the design objective under substantially weaker assumptions than those in \eqref{eq:rank_con}. Taking advantage of the geometric approach, our methodology is coordinate-free and provides a unified framework for both continuous-time and discrete-time systems. The key contributions are summarized as follows.
\begin{enumerate}
    \item We present a geometric approach for estimating a linear projection of the system state in the presence of unknown inputs, effectively reducing the state dimension. This projection is guaranteed to have the smallest kernel (null space), thereby capturing the maximum possible information about the system state. Furthermore, we analyze the performance of our distributed estimation scheme in systems with zero dynamics—a particularly challenging and non-trivial case for DUIO design.
       \item Building on the previous point, we develop a novel DUIO for continuous-time systems. Unlike existing DUIOs \cite{yang2022state,cao2023distributed,cao2023distributedauto,zhu2024distributed,disaro2024distributed}, the proposed approach eliminates the reliance on the individual input-output rank condition \eqref{eq:rank_con} and instead employs weaker joint conditions across all sensor nodes, thereby allowing the unknown inputs to be addressed collectively by all sensors. Moreover, the sufficiency and necessity of this condition for DUIO design are rigorously characterized, representing a significant advancement over the results in \cite{cao2023bounded}.
    \item We further propose a two-time scale framework for DUIO design in the discrete-time domain, employing separate local adaptation and communication iterations. To the best of our knowledge, this is the first work to address discrete-time DUIOs without imposing rank constraints \eqref{eq:rank_con}. The proposed solution is proven to achieve Q-linear convergence with the fastest attainable rate. In addition, our analysis reveals how steady-state accuracy depends on communication frequency, consistent with the findings of \cite{kim2024design}, which highlight the role of implementing large coupling gains in extending continuous-time designs to the discrete-time setting.
\end{enumerate}

The paper is organized as follows. 
Section~\ref{sec:Prelim} introduces key notations and definitions and formulates the problem.
In Section~\ref{sec:geometric_approach}, we describe the geometric approach used for subspace decomposition, which forms the foundation of our DUIO designs. Subsequently, Section~\ref{sec:DUIO} proposes the DUIO schemes for both continuous- and discrete-time LTI systems. Numerical simulations and a case study of a distributed generation unit (DGU) scenario from electric power systems are provided in Section~\ref{sec:simulation} to illustrate the effectiveness of the proposed distributed scheme. Finally, Section~\ref{sec:conclusion} concludes the paper and suggests potential directions for future research.

\section{Preliminaries and Problem Formulation}\label{sec:Prelim}
\subsection{Notations}
Let $\mathbb{R}$ and $\mathbb{C}$ denote the real and complex fields, respectively. $\R_{>0}$ is the set of positive real numbers. A symmetric partition in $\mathbb{C}$ is denoted by $\C = \C_g \cup \C_b$ with $\C_g \cap \C_b = \varnothing$ where $\C_g$ denotes the partial complex plane consisting of ``good" (e.g. asymptotically stable) eigenvalues and $\C_b$ denotes the partial complex plane consisting of ``bad" (e.g., marginally stable or unstable) eigenvalues. We denote by $I_n$ the $n\times n$ identity matrix. Similarly, $\mathbf{1}_{n\times n}$ and 
$\mathbf{0}_{n\times n}$ denote the $n \times n$ all-ones matrix and the $n \times n$ zero matrix, respectively.
For a vector $v$, $\|v\|_1$, $\|v\|_2$ and $\|v\|_\infty$ denote the 1-norm, 2-norm and $\infty$-norm, respectively. For a matrix $M$, $\|M\|_1$ and $\|M\|_\infty$ denote the induced matrix 1-norm and 
$\infty$-norm, respectively. $\otimes$ stands for the Kronecker product, and $\uplus$ represents the multiset union, where common elements are retained as duplicates. $\sign(\cdot)$ is the sign function, which is defined as $\sign(x)=1$ if $x>0$, $\sign(x)=-1$ if $x<0$, and $\sign(x)=0$ if $x=0$.
$\mathrm{col}(M_1,M_2,\ldots,M_n)$ denotes the stacked matrix $[M_1^\top,M_2^\top,\cdots,M_n^\top]^\top$ and ${\rm diag}(M_1,M_2,\ldots,M_n)$ denotes the block diagonal matrix composed of $M$'s. $M^\dagger$ represents the pseudo inverse of matrix $M$. $\kappa(M)$ represents the spectrum of $M$, and $\lambda_i(M)$, $\lambda_{\max}(M)$ represent its $i$-th and maximum eigenvalue. $\sigma_{\min}(M)$ is the minimum singular value of $M$.

\subsection{Problem setup}
Recall the system given in \eqref{eq:systemDecomposition} and the sensor measurements in \eqref{eq:output}.
The information exchange in the sensor network considered in this article is characterized by an undirected communication graph $\mathcal{G}=(\mathbf{N},\mathcal{E},\mathcal{A})$, where $\mathbf{N}$ is the set of sensor nodes, $\mathcal{E}\subseteq \mathbf{N}\times \mathbf{N}$ denotes the set of communication links between nodes. $\mathcal{A}=[a_{ij}]\in\mathbb{R}^{N\times N}$ is the adjacency matrix, where if $(i,j)\in\mathcal{E}$, $a_{ij}=a_{ji}=1$, and  $a_{ij}=a_{ji}=0$ otherwise. 
The Laplacian of $\mathcal{G}$ is the matrix $\mathcal{L}=[l_{ii}]\in{\mathbb{R}^{n\times n}}$ with $l_{ij} = \sum_{j=1,j\neq i}^N a_{ij},\,\,i=j$ and $l_{ij} = -a_{ij},\,\,i\neq j$.
The following assumption is adopted in this article and is commonly used in the literature \cite{yang2022state, cao2023distributed}.

\begin{assumption}\label{as:connected}
    The communication topology $\mathcal{G}=\{\mathbf{N},\mathcal{E},\mathcal{A}\}$ associated with the network of DUIO is connected.
\end{assumption}
Under Assumption~\ref{as:connected}, the Laplacian matrix has a key property, as stated in the following Lemma. 
\begin{lemma}\cite{ren2007information}\label{lemma:eig_Laplacian_1}
Given an undirected connected graph $\mathcal{G}=(\mathbf{N},\mathcal{E},\mathcal{A})$, $\frac{\mathbf{1}_{N\times 1}}{\sqrt{N}}$ is the unique left and right eigenvector of its Laplacian matrix $\mathcal{L}$ corresponding to the zero eigenvalue.
\end{lemma}

Let $\hat x_i(t)$ denote the local state estimate at node $i$. The corresponding estimation error is defined as
\begin{equation}
    \label{eq:error}
    e_i(t) \coloneq x(t) - \hat x_i(t) \, .
\end{equation}
Accordingly, the global estimation error can be defined by
\begin{equation}\label{eq:global_error}
    e(t) \coloneq \col(e_1(t), e_1(t), \cdots, e_N(t)) \, .
\end{equation}
The goal is to design a set of networked observers $\{ \mathcal O_i \}_{i \in \mathbf N}$ under a communication graph that satisfies Assumption~\ref{as:connected}, with guaranteed convergence properties without imposing \eqref{eq:rank_con}. Specifically, the design ensures
\begin{equation}\label{def:dUIO_ct}
        \lim_{t\rightarrow +\infty} \|e(t)\| = 0 \, ,
    \end{equation}
   for the continuous-time system \eqref{eq:systemDecomposition_ct}, or
    \begin{equation}\label{def:dUIO_dt}
         \lim_{t\rightarrow +\infty} \|e(t)\| \leq \alpha(\|x(t)\|,d) \, ,
    \end{equation}
    for the discrete-time system \eqref{eq:systemDecomposition_dt}, where $\alpha(\|x(t)\|,d)$ is a class $\mathcal{KL}$ function that depends on the state norm $\|x(t)\|$ and the number of communication rounds $d$ between two consecutive time steps.

\section{Geometric Approach based on Subspace Decomposition}\label{sec:geometric_approach}
\subsection{Basic Elements
}

In this subsection, we briefly review fundamental concepts related to the geometric approach for completeness. We adopt the terminology and notation style established in \cite{massoumnia1986geometric,wonham1985linear}. Readers already familiar with this material, or primarily interested in the observer design, may choose to skip this subsection and revisit it later as needed.

\subsubsection{Notations and Definitions}
Let $A:\X \rightarrow \X$ be an endomorphism, and let $\V \subseteq \X$ be a subspace with an insertion map $W:\W \rightarrow \X$, i.e., $\W=\image W$ and $W$ is monic. A subspace $\W \subseteq\X$ is said to be {\em invariant} with respect to a map $A:\X\rightarrow \X$ if $A\W \subseteq \W$. For an invariant subspace $\W$, we denote by $A|\W: \W \to \W$ the {\em restriction of $A$ to $\W$}, i.e., the unique map satisfying $AW=W(A|\W)$, where $W:\W \rightarrow \X$ is the insertion map of $\W$ in $\X$.
The set $\mathcal{W}=x+\W$ is called a coset of $\W$ in $\X$, and $x\in \X$ is a representative for $\W$. The set of all cosets of $\W$ in $\X$ is denoted by
$ \X/\W\coloneqq\{x+\W: x\in \X\} $
and is called the quotient (factor) space of $\X$ modulo $\W$. 
Sometimes we denote $\X/\W$ as $\frac{\X}{\W}$ for simplicity. 
We denote by $A|\X/\W$ the map {\em induced on $\X/\W$ by $A$}, satisfying $\left( A|\X \!/\! \W \right) P=PA$, where $P:\X\to\X/\W$ is the canonical projection on $\X/\W$. Let $C:\X \rightarrow \Y$ be a map. If $\Ss \subseteq \Y$, $C^{-1}\Ss$ denotes the \textit{inverse} image of $\Ss$ under $C$, which is defined by
    $C^{-1}\Ss \coloneqq \{x:x\in \X \ \& \  Cx \in \Ss\} \subseteq \X$.
Note that $C^{-1}$ is the inverse image function of the map $C$, and as such it is regarded as a function from the set of all subspaces of $\Y$ to those of $\X$. If $\mathscr{R}, \mathscr{S} \subseteq \mathscr{X}$, we define the subspace $\mathscr{R} + \mathscr{S} \subseteq \mathscr{X}$ as $\mathscr{R} + \mathscr{S}= \{r+s: r\in \mathscr{R}\ \&\  s\in \mathscr{S}\}$, and we define the subspace $\mathscr{R} \cap \mathscr{S} \subseteq \mathscr{X}$ as $\mathscr{R} \cap \mathscr{S}= \{x: x\in \mathscr{R}\ \&\  x\in \mathscr{S}\}$. 
The symbol $\oplus$ indicates that the subspaces being added are independent.
We indicate that two vector spaces $\mathscr V$ and $\mathscr W$ are isomorphic by $\mathscr V \simeq \mathscr W$.
Furthermore, $\mathrm{Mat}(A)$ is the matrix representation of the map $A$ relative to the given basis pair.
Let $A:\X \to \X$ and $\K \subset \X$ be a map and subspace, we define $\ang{\K \, |\,A }$ as
$
\ang{\K \,|\,A } \coloneq \K \cap A^{-1} \K \cap \ldots \cap A^{-n+1} \K.
$
 
\subsubsection{$(C,A)$-invariant Subspace}
Let $A:\X\rightarrow \X$ and $C:\X\rightarrow \Y$. We say a subspace $\W \subseteq \X$ is $(C,A)$-invariant if there exists a map $L:\Y\rightarrow \X$ such that
\begin{equation}\label{eq:def(C,A)-inv}
    (A+LC)\W \subseteq \W.
\end{equation}

In some literature, $(C,A)$-invariant subspace is also called a conditioned invariant subspace \cite{basile1992controlled}. The class of $L$ for which \eqref{eq:def(C,A)-inv} holds is denoted by $\mathbf{L}(\W)$ if $A$ and $C$ are clear in the context. The notation $L\in \mathbf{L}(\W)$ reads ``$L$ is a friend of $\W$". 

Let $\B\subseteq \X$. We write $\underline{\W}(C,A;\B)$ to denote the class of $(C,A)$-invariant subspaces containing $\B$. $\underline{\W}(C,A;\B)$ is closed under the operation of subspace intersection; hence, $\underline{\W}(C,A;\B)$ has an infimal element (i.e., a subspace with minimal dimension), which is denoted by $\W^*(C,A;\B)$ (the infimal $(C,A)$-invariant subspace containing $\B$) or simply $\W^*(\B)$ if $A$ and $C$ are clear in the context.

\subsubsection{Unobservability Subspace}
We say a subspace $\Ss \subseteq \X$ is an unobservability subspace if 
\begin{equation}\label{eq:def_Ss}
    \Ss = \ang{\kernel HC\,|\,A+LC }
\end{equation}
for some output injection map $L:\Y \rightarrow \X$ and measurement mixing map $H:\Y \rightarrow \Y$. Note that the \textit{unobservability subspace} is a different concept from the \textit{unobservable subspace}\footnote{The unobservable subspace can be defined as $\ang{\kernel C \,|\, A}\coloneqq \kernel C \cap  A^{-1} \kernel C \cap A^{-2} \kernel C \cap \cdots \cap A^{-n+1} \kernel C$.}.

The class of $L$ for which \eqref{eq:def_Ss} holds is denoted by $\mathbf{L}(\Ss)$ if $A$ and $C$ are clear in the context. The notation $L\in \mathbf{L}(\Ss)$ reads ``$L$ is a friend of $\Ss$". 

Let $\B\subseteq \X$. We write $\underline{\Ss}(C,A;\B)$ to denote the class of unobservability subspaces containing $\B$. $\underline{\Ss}(C,A;\B)$ has an infimal element, which is denoted by $\Ss^*(C,A;\B)$ (the infimal unobservability subspace containing $\B$) or simply $\Ss^*(\B)$ if $A$ and $C$ are clear in the context.

Note that the concepts of $(C, A)$-invariant subspaces and the unobservability subspace are directly useful for the subspace decomposition introduced in the next subsection, which plays a key role in our DUIO designs in \eqref{eq:DUIO_ct} and \eqref{eq:dis_DUIO}.

\subsection{Subspace Decomposition}\label{sec:space_decomposition}
It is crucial to isolate the largest possible subspace of the state that remains unaffected by unknown inputs, i.e., to maximize the portion of state information that is immune to their influence.
To avoid invoking the restrictive assumption on \eqref{eq:rank_con}, we propose a method for this extraction by decomposing the system \eqref{eq:output} and \eqref{eq:systemDecomposition} at each node using a geometric approach, which is applicable to both continuous and discrete-time systems, as will be outlined below.

The key idea is to identify, at each node, the infimal $(C_i, A)$-invariant subspace $\W_{g,i}^*$ that contains $\image \bar{B}_i$, i.e., the component of the state that is influenced by the unknown input, meanwhile allowing us to assign the spectrum of the map induced on $\X/\W_{g,i}^*$ by $A_{L_i}$\footnote{To simplify the notation, we define $A_{L_i} \coloneqq A+L_iC_i$, where $L_i: \Y_i \to \X$ is the output injection map at node~$i$.} i.e., $A_{L_i}|\X/\W_{g,i}^*$ into the good
part of the complex plane by selecting an appropriate output injection map $L_i$. Consequently, we can estimate $P_{W_{g,i}^*} x$ without being affected by $\bar u_i$, where $P_{W_{g,i}^*} : \X \to \X/\W_{g,i}^*$ is a canonical projection that minimizes the loss of state information.
In other words, the dimension of $\kernel P_{W_{g,i}^*}$, which represents the portion of state information discarded at node $i$, should be minimized.

Now, let us begin with the identification of the subspace $\W_{g,i}^*$, which plays a key role in our method. Consider a family of subspaces that satisfy our requirements
\begin{equation}\label{eq:def_W_gfamily}
\begin{split}
    \underline{\W_{g,i}} 
    \coloneqq 
    \{&\W_i: \W_i \in \underline{\W}(C_i,A;\image \bar{B}_i),\\ 
    &\ \& \,\exists L_i \in \mathbf{L}(\W_i)\ \mathrm{s.t.}\ \kappa(A_{L_i}|\X/\W_i) \subset \C_g \}.
\end{split}
\end{equation}
It is straightforward to verify that all the subspaces in the set $\underline{\mathscr{W}_{g,i}}$ are $(C_i, A)$-invariant while containing $\image \bar{B}_i$. Moreover, the corresponding induced map of $A_{L_i}$ of these subspaces can be made stable by appropriately choosing the output injection map $L_i$.
Since $\underline{\W_{g,i}}$ is closed under intersection \cite[Lemma~7]{massoumnia1986geometric}, it follows immediately that $\underline{\W_{g,i}}$ contains an \textit{infimal element} $\W_{g,i}^* \coloneqq \inf  \underline{\W_{g,i}}$, which is our desired subspace to identify. 
Now it is critical to understand whether there exists an $L_i$ such that the spectrum of $A_{L_i}|\X/\W_i$ can be assigned to lie within the stable region.
This is addressed by the following lemma.
\begin{lemma}\label{lemma:spectrumfixed&free} \cite[Proposition~20]{massoumnia1986geometric}
    Let $\W_i\in\underline{\W}(C_i,A;\image \bar{B}_i)$ and $\Ss_i^* \coloneqq \inf \underline{\Ss}(C_i,A;\image \bar{B}_i)$. Given $L_i\in \mathbf{L}(\W_i)$, then
    \begin{equation}
        \kappa( A_{L_i}|\X/\W_i) = \kappa_i \uplus \kappa_{z,i}
    \end{equation}
    where $\kappa_i\coloneqq\kappa(A_{L_i}|\X/\Ss_i^*)$ is \textbf{freely assignable} by a suitable choice of $L_i \in \mathbf{L}(\W_i)$, and $\kappa_{z,i}\coloneqq\kappa(A_{L_i}|\Ss_i^*/\W_i)$ is \textbf{fixed} for all $L_i\in \mathbf{L}(\W_i)$.
\end{lemma}
Clearly, $\W_{g,i}^*\in\underline{\W_{g,i}}(C_i,A;\image \bar{B}_i)$. Lemma~\ref{lemma:spectrumfixed&free} shows that for any choice of $L_i \in \mathbf{L}(\W^*_{g,i})$, a portion of the spectrum of $A_{L_i}|\Ss_i^*/\W^*_{g,i}$, denoted by $\kappa(A_{L_i}|\Ss_i^*/\W_{g,i}^*)$, remains fixed and cannot be modified by $L_i \in \mathbf{L}(\W^*_{g,i})$. Therefore, $\kappa(A_{L_i}|\Ss_i^*/\W_{g,i}^*)$ is desired to be ``good'' for all $L_i \in \mathbf{L}(\W^*_{g,i})$. The decomposition of the induced map $A_{L_i}|\X/\W_{g,i}^*$ is illustrated in Fig.~\ref{fig:good_decompose} \cite[Chapter 0.6]{wonham1985linear}, where $\kappa( A_{L_i}|\Ss_i^*/\W_{g,i}^*)$ is fixed for any $L_i \in \mathbf{L}(\W_i^*)$ but with ``good'' spectrum, $\kappa( A_{L_i}|\X/\Ss_i^*)$ can be freely assigned by a suitable choice of $L_i \in \mathbf{L}(\W_i^*)$.
\begin{figure}[htp]
    \centering
    \scalebox{0.9}{

\begin{tikzpicture}


\def\xgap{5}
\def\ygap{1.5}
\foreach \i in {0,...,3}{
    \foreach \j in {0,...,3}{
        \coordinate (\i\j) at (\i*\xgap, \j*\ygap);
    }
}
\pgfmathatantwo{\ygap}{\xgap}
\def\labangle{\pgfmathresult}

\tikzset{label/.style = {font=\small, midway}}


\node (XmodS) at (00) {$\X / \Ss^*_i$};
\node (XmodS2) at (10) {$\X / \Ss^*_i$};
\node (XmodWg) at (01) {$\X / \W^*_{g,i}$};
\node (XmodWg2) at (11) {$\X / \W^*_{g,i}$};
\node (SmodWg) at (02) {$\Ss^*_i / \W^*_{g,i}$};
\node (SmodWg2) at (12) {$\Ss^*_i / \W^*_{g,i}$};

\draw[-latex] (XmodS) -- (XmodS2) node [label, above] {$A_{L_i}|\X/\Ss_i^*$};
\draw[-latex] (XmodS) -- (XmodS2) node [label, below] {spectrum free};

\draw[-latex] (XmodWg) -- (XmodWg2) node [label, above] {$A_{L_i}|\X/\W_{g,i}^*$};
\draw[-latex] (XmodWg) -- (XmodWg2) node [label, below] {spectrum good};

\draw[-latex] (SmodWg) -- (SmodWg2) node [label, above] {$A_{L_i}|\Ss^*_i/\W_{g,i}^*$};
\draw[-latex] (SmodWg) -- (SmodWg2) node [label, below] {spectrum fixed \& good};

\draw[-latex] (SmodWg) -- (XmodWg) node [label, right] {$\bar S_{g,i}$};
\draw[-latex] (SmodWg2) -- (XmodWg2) node [label, right] {$\bar S_{g,i}$};

\draw[-latex] (XmodWg) -- (XmodS) node [label, right] {$P_{W_{g,i}^*S_i^*}$};
\draw[-latex] (XmodWg2) -- (XmodS2) node [label, right] {$P_{W_{g,i}^*S_i^*}$};

\end{tikzpicture}}\\[-1.2ex]
    \caption{Commutative diagram of the restriction and induced map of $ A_{L_i}|\X/\W_{g,i}^*$. 
    $\bar{S}_{g,i}$ is defined as the insertion map $\bar{S}_{g,i}:\Ss_i^*/\W_{g,i}^*\rightarrow\X/\W_{g,i}^*$, where $\Ss_i^*/\W_{g,i}^*\simeq\frac{\X/\W_{g,i}^*}{\X/\Ss_i^*}$. $P_{W_{g,i}^*S_i^*}$ is defined as the canonical projection $P_{W_{g,i}^*S_i^*}:\X/\W_{g,i}^*\rightarrow\X/\Ss_{i}^*$ such that this diagram commutes.}
    \label{fig:good_decompose}
\end{figure}
\begin{remark}\label{rem:up-low:W_g}
    Let $\W_i^*\coloneqq\inf \underline{\W}(C_i,A;\image \bar{B}_i)$. It can be inferred from Lemma~\ref{lemma:spectrumfixed&free} that 
    \begin{align*}
        \kappa(A_{L_i}|\X/\W_i^*) = \kappa(A_{L_i}|\X/\Ss_i^*) \uplus \kappa(A_{L_i}|\Ss_i^*/\W_i^*)       
    \end{align*}
    and $\kappa(A_{L_i}|\Ss_i^*/\W_{i}^*)$ is independent of the choice of $L_i \in \mathbf{L}(\W_i^*)$. As the distribution of $\kappa(A_{L_i}|\Ss_i^*/\W_{i}^*)$ depends only on the system parameters and may not be entirely ``good'', unless $\W_{g,i}^* = \W_i^*$ there may not exist $L_i \in \mathbf{L}(\W_i^*)$ such that $\kappa(A_{L_i}|\X/\W_i^*) \subset \C_g$. However, if one takes $\W_{g,i}^* = \Ss_i^*$, then $\kappa(A_{L_i}|\Ss_i^*/\W_{i}^*)$ becomes free and can always be made ``good'' by a suitable choice of $L_i$. However, this choice of $\W_{g,i}^*$ is not necessarily infimal, since it may discard some ``good'' spectral elements of $\kappa(A_{L_i}|\Ss_i^*/\W_{i}^*)$ by enforcing $\W_{g,i}^* = \Ss_i^*$.
    While finding a precise $\W_{g,i}^*$ is still under investigation, one can easily conclude that
    \begin{equation*}
        \W^*_i \subseteq \W^*_{g,i} \subseteq \Ss^*_i
    \end{equation*}
    where the two set equality signs correspond to the cases of $\kappa(A_{L_i}|\Ss_i^*/\W_i^*)$ being entirely ``good'' and entirely ``bad'', respectively.
\end{remark}
As detailed in Remark~\ref{rem:up-low:W_g}, we proceed to identify and isolate the ``good" and ``bad" modes of $\kappa(A_{L_i}|\Ss_i^*/\W_i^*)$ as they are invariant with respect to $L_i$.
Let $\beta_i(\lambda)$ be the minimal polynomial of $A_{L_i}|\Ss_i^{*}/\W_i^*$. Factor $\beta_i(\lambda)=\beta_{g,i}(\lambda)\beta_{b,i}(\lambda)$, where the zeros of $\beta_{g,i}$ in $\C$ belong to $\C_g$ and the zeros of $\beta_{b,i}(\lambda)$ in $\C$ belong to $\C_b$. Let
\begin{equation*}
\begin{aligned}
    &\bar{\X}_{g,i}^* \coloneqq \frac{\Ss_i^{*}}{\W_i^{*}}\  \bigcap\  \kernel \beta_{g,i}(A_{L_i}|\Ss_i^*/\W_i^*), \\
    &\bar{\X}_{b,i}^* \coloneqq \frac{\Ss_i^{*}}{\W_i^{*}}\  \bigcap \ \kernel \beta_{b,i}(A_{L_i}|\Ss_i^*/\W_i^*).
\end{aligned}
\end{equation*}
Moreover, $\beta_{g,i}$ and $\beta_{b,i}$ coprime implies that
\begin{equation}\label{eq:Xa_Xb}
    \frac{\Ss_i^{*}}{\W_i^{*}} = \bar{\X}_{g,i}^*\oplus \bar{\X}_{b,i}^*.
\end{equation}
which ``split'' $\Ss_i^*/\W_i^*$ into two sub-quotient subspaces associated with the ``good'' and ``bad'' modes. Now, we introduce the main theorem for solving $\W_{g,i}^*$.
\begin{theorem}\label{thm:Wg}
    Let $P_{W_i^*}:\X \to \X/\W_i^*$ be the canonical projection, where $\W_i^*$ is the infimal $(C_i,A)$-invariant subspace.
    Then, the subspace $\W_{g,i}^*$, defined as
    \begin{equation}\label{eq:Wg}
        \W_{g,i}^*
        \coloneqq 
        P_{W_i^*}^{-1}\bar{\X}_{b,i}^* .
    \end{equation}
    is the infimal element of the family $\underline{\W_{g,i}}$ defined by \eqref{eq:def_W_gfamily}.
\end{theorem}
\begin{proof}
    First, we check that $\W_{g,i}^*\in\underline{\W_{g,i}}$. From \eqref{eq:Wg}, it is easy to conclude
    \begin{equation*}
        \W_{g,i}^*\supseteq \W_{i}^* \supseteq \image{\bar{B}_i}.
    \end{equation*}
Moreover, since
\begin{equation*}
\begin{aligned}
    P_{W_i^*}A_{L_i}\W_{g,i}^* &= \left(A_{L_i}|\X/\W_i^*\right)P_{W_i^*}\left(P_{W_i^*}^{-1}\bar{\X}_{b,i}^*\right)\\
    &\subseteq \left(A_{L_i}|\X/\W_i^*\right)\bar{\X}_{b,i}^*\subseteq \bar{\X}_{b,i}^*,
\end{aligned}
\end{equation*}
one obtains
\begin{equation*}
    A_{L_i}\W_{g,i}^* \subseteq P_{W_i^*}^{-1}\bar{\X}_{b,i}^* =\W_{g,i}^*,
\end{equation*}
which indicates $\W_{g,i}^*$ is $A_{L_i}$-invariant.
Note that
\begin{equation*}
    \frac{\Ss_i^*}{\W_{g,i}^*} = \frac{\Ss_i^*/\W_i^*}{\bar \X_{b,i}^*} = \bar \X_{g,i}^*.
\end{equation*}
Then, based on Lemma~\ref{lemma:spectrumfixed&free}, we have
\begin{align*}\label{eq:obj}
    \kappa(A_{L_i}|\X/\W_{g,i}^*) 
    &= 
    \kappa(A_{L_i}|\X/\Ss_i^*)\uplus \kappa(A_{L_i}|\Ss_i^*/\W_{g,i}^*) \\
    &=
    \kappa(A_{L_i}|\X/\Ss_i^*)\uplus \kappa(A_{L_i}|\bar{\X}_{g,i}^*)
\end{align*}
which ensures the existence of an $L_i$ such that $\kappa(A_{L_i}|\X/\W_{g,i}^*)\subset \mathbb{C}_g$. Therefore, we confirm $\W_{g,i}^* \in \underline{\W_{g,i}}$.

We now prove that $\W_{g,i}^*$ is the infimal element of $\underline{\W_{g,i}}$. Consider an arbitrary $\W_i^{(\mathrm{a})}\in\underline{\W_{g,i}}$. Since $\W_i^{(\mathrm{a})} \supseteq \W_i^*$, it follows from \eqref{eq:Xa_Xb} that we can  decompose $P_{W_i^*}\W_i^{(\mathrm{a})}$ as
\begin{equation*}
\begin{aligned}
        P_{W_i^*}\W_i^{(\mathrm{a})} = &\left(\left(P_{W_i^*}\W_i^{(\mathrm{a})}\right)\cap \bar{\X}_{b,i}^*\right) \\&\oplus \left(\left(P_{W_i^*}\W_i^{(\mathrm{a})}\right)\cap \bar{\X}_{g,i}^*\right).
\end{aligned}
\end{equation*}
Since $\W_i^{(\mathrm{a})}\in\underline{\W_{g,i}}$, we must have \[\left(P_{W_i^*}\W_i^{(\mathrm{a})}\right)\cap \bar{\X}_{b,i}^*=\bar{\X}_{b,i}^*.\]
Hence, 
    $P_{W_i^*}\W_i^{(\mathrm{a})}\supseteq \bar{\X}_{b,i}^*$.
It follows immediately that
\begin{equation*}
    \W_i^{(\mathrm{a})}\supseteq P_{W_i^*}^{-1} \bar{\X}_{b,i}^* = \W_{g,i}^*.
\end{equation*}
Since every $\W_i^{(\mathrm{a})}\in\underline{\W_{g,i}}$ contains $\W_{g,i}^*$, we conclude that $\W_{g,i}^*$ is infimal, which completes the proof.
\end{proof}

\begin{remark}
    Theorem~\ref{thm:Wg} provides the definition of the $\W_{g,i}^*$ from a geometric point of view. To understand it better, we can view \eqref{eq:Wg} from two extreme scenarios. If all the elements in $\kappa(A_{L_i}|\Ss_i^*/\W_{i}^*)$ are ``good'', then $\bar \X_{b,i}^* = 0$, which yields $\W_{g,i}^* = P_{W_i^*}^{-1} 0 = \W_i^*$ \cite[Chapter~0.12]{wonham1985linear}. On the other side, if $\kappa(A_{L_i}|\Ss_i^*/\W_{i}^*)$ is entirely ``bad'', we have $\bar \X_{b,i}^* = \Ss_i^*/\W_i^*$, which implies $\W_{g,i}^* = P_{W_i^*}^{-1} \left(\Ss_i^*/\W_i^*\right) = \Ss_i^*$. These two scenarios coincide with the assertion in Remark~\ref{rem:up-low:W_g}. Furthermore, Fig.\ref{fig:lattice_W_g} illustrates the inclusion relations among the subspaces associated with Theorem\ref{thm:Wg}.
\end{remark}
\begin{figure}[htp]
    \centering
    \scalebox{0.9}{
    \begin{tikzpicture}

\draw(0,0) -- (0,5);

\filldraw [black] (0,5) circle (2pt);
\node (n1) at (-0.5,5) {$\X$};

\filldraw [black] (0,4) circle (2pt);
\node (n2) at (-0.5,4) {$\Ss_i^*$};

\filldraw [black] (0,3) circle (2pt);
\node (n3) at (-0.5,3) {$\W_{g,i}^*$};

\filldraw [black] (0,2) circle (2pt);
\node (n4) at (-0.5,2) {$\W_i^*$};

\filldraw [black] (0,1) circle (2pt);
\node (n5) at (-0.5,1) {$\image{\Bar{B}_i}$};

\filldraw [black] (0,0) circle (2pt);
\node (n6) at (-0.5,0) {$0$};

\draw [decorate,decoration={brace,amplitude=7pt,mirror,raise=4pt},yshift=0pt]
(0.1,2) -- (0.1,4) node [black,midway,xshift=0.5cm] {};

\draw(0.55,3) -- (1,3.25);
\node [rotate=120] (f1) at (1.1,3.3) {$\approx$};
\draw(1.2,3.35) -- (2,3.77);

\filldraw [black] (2,3.77) circle (2pt);
\node (n7) at (3,3.77) {$\Ss^*_i/\W_i^{*}$};

\draw(2,3.77)--(1,2.5);
\filldraw [black] (1,2.5) circle (2pt);
\node (n8) at (1,2.1) {$\bar \X_{b,i}^*$};

\draw(2,3.77)--(3,2.5);
\filldraw [black] (3,2.5) circle (2pt);
\node (n9) at (4,2.1) {$\bar \X_{g,i}^* = \Ss_i^*/\W_{g,i}^*$};

\end{tikzpicture}}\\[-1.5ex]
    \caption{Lattice diagrams: Construction of the subspace $\W_{g,i}^*$}
    \label{fig:lattice_W_g}
\end{figure}
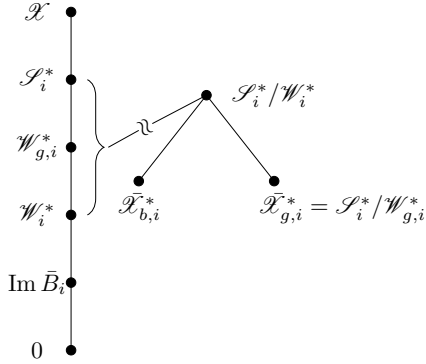
\begin{figure}[htp] 
\centering
    \scalebox{0.9}{
    \begin{tikzpicture}


\def\xgap{3.25}
\def\ygap{1.5}
\foreach \i in {0,...,3}{
    \foreach \j in {0,...,3}{
        \coordinate (\i\j) at (\i*\xgap, \j*\ygap);
    }
}
\pgfmathatantwo{\ygap}{\xgap}
\def\labangle{\pgfmathresult}

\tikzset{label/.style = {font=\small, midway}}

\node (barU) at (01) {$\bar \U_i$};
\node (XmodWg) at (10) {$\X / \W^*_{g,i}$};
\node (XmodWg2) at (20) {$\X / \W^*_{g,i}$};
\node (X) at (11) {$\X$};
\node (X2) at (21) {$\X$};
\node (Wg) at (12) {$\W^*_{g,i}$};
\node (Wg2) at (22) {$\W^*_{g,i}$};

\draw[-latex, dashed] (barU) -- (XmodWg) node [label, below, rotate=-\labangle] {$0$};
\draw[-latex] (barU) -- (X) node [label, above] {$\bar B_i$};

\draw[-latex] (XmodWg) -- (XmodWg2) node [label, above] {$A_{L_i}|\X/\W_{g,i}^*$};
\draw[-latex] (XmodWg) -- (XmodWg2) node [label, below] {spectrum good};
\draw[-latex] (X) -- (X2) node [label, above] {$A_{L_i}$};
\draw[-latex] (Wg) -- (Wg2) node [label, above] {$A_{L_i}|\W_{g,i}^*$};

\draw[-latex] (Wg) -- (X) node [label, right] {$W_{g,i}^*$};
\draw[-latex] (Wg2) -- (X2) node [label, right] {$W_{g,i}^*$};

\draw[-latex] (X) -- (XmodWg) node [label, right] {$P_{W_{g,i}^*}$};
\draw[-latex] (X2) -- (XmodWg2) node [label, right] {$P_{W_{g,i}^*}$};

\end{tikzpicture}}\\[-1.5ex]
    \caption{Commutative diagram of $\W_{g,i}^*$ decomposition at node $i$.}
    \label{fig:Wg_decompose}
\end{figure}
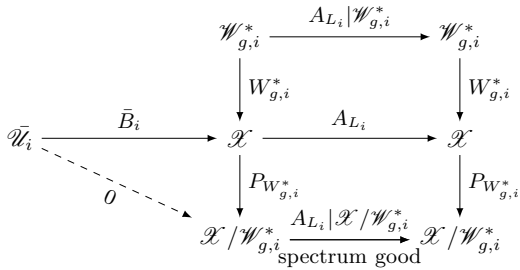
Having identified the desired subspace $\W_{g,i}^*$, we can now easily reconstruct $P_{W_{g,i}^*}x$ at each node. This reconstruction is immune to the influence of the unknown input $\bar{u}_i$, as shown in the commutative diagram in Fig.~\ref{fig:Wg_decompose}. 
\begin{remark}
The idea of partial state estimation has also been investigated in \cite{tsui1996new} and \cite{sundaram2008partial}, which, to some extent, resemble the DUIO framework in that individual nodes are unable to estimate the full system state. {However, \cite{tsui1996new} lacks a comprehensive analysis of the system's invariant zeros, leaving the maximum number of recoverable system states undefined.  In \cite{sundaram2008partial}, the maximum number of reconstructible state combinations is determined through a sequence of algebraic parameterizations arising in the observer design procedure. In contrast, our geometric analysis provides an a priori characterization of the reconstructible state subspace directly from the intrinsic structure of the system matrices, independent of any particular observer synthesis.}
\end{remark}

Note that, as established in \cite[Lemma~4, Theorem~19]{massoumnia1986geometric}, the existence of subspaces $\W_i^*$ and $\Ss_i^*$ is independent of the output injection map $L_i$. Furthermore, according to Theorem~\ref{thm:Wg}, the definition of subspace $\W_{g,i}^*$ is not related to any specific choice of $L_i$. Instead, it is defined by the inverse map $P_{W_i^*}^{-1}$ and the partial invariant zeros $\bar{\X}_{b,i}^*$. Consequently, the analysis of the subspace decomposition is performed without choosing a specific $L_i$.
For implementation, we provide the following algorithms to compute the basis of $\W_{g,i}^*$ and the corresponding output feedback gains $L_i$. Algorithm~\ref{alg:W_i&S_i} first outlines the procedure to obtain the bases of $\W_i^*$ and $\Ss_i^*$. Building on these results, Algorithm~\ref{alg:W_gi} computes the basis of $\W_{g,i}^*$ together with the matrix representation of the canonical projection $P_{W_{g,i}^*}$. Finally, Algorithm~\ref{alg:Li} presents the numerical calculation of the output injection map $L_i$.
Note that the computations in Algorithms~\ref{alg:W_i&S_i}-~\ref{alg:Li} are performed entirely offline for each node, and therefore impose no additional online computational burden during the state estimation process.

%

%
\begin{algorithm}[ht] \label{alg:W_i&S_i}
    \caption{Numerical calculation of the basis of $\W_i^*$, $\Ss_i^*$ and the canonical projection $P_{W_i^*}$, $P_{S_i^*}$ at node $i$}
    \LinesNumbered
    \KwIn{
        System matrices $A$, $\bar B_i$, $C_i$.}
        $W_i^0 \gets \mathbf{0}_{n\times n}$\\
        \For{$k\gets 1$ \KwTo $n$}{
            $P_i^k \gets \pr{\kernel {W_i^0}^\top}^\top $ \\
            $T_i^k \gets \kernel 
                    \begin{bmatrix}
                        P_i^k \\ C_i
                    \end{bmatrix}$ \\
            $W_i^k \gets
                    \begin{bmatrix}
                        \bar B_i & AT_i^k 
                    \end{bmatrix}$
        }
        $P_{W_i^*} \gets P_i^n$ $\left(P_{W_i^*}\in\mathbb{R}^{(n-w_i^*)\times n}\right)$\\
        $W_i^* \gets W_i^n$ $\left(W_i^*\in\mathbb{R}^{n\times w_i^*}\right)$\\
        $S_i^0 \gets I_n$\\
        \For{$k\gets 1$ \KwTo $n$}{
            $P_i^k \gets \pr{\kernel {S_i^0}^\top}^\top $ \\
            $T_i^k \gets \kernel 
                    \begin{bmatrix}
                        P_i^k A \\ C_i
                    \end{bmatrix}$ \\
            $S^k \gets
                    \begin{bmatrix}
                        W_i^* & T_i^k 
                    \end{bmatrix}$
        }
        $P_{S_i^*} \gets P_i^n$ $\left(P_{S_i^*}\in\mathbb{R}^{(n-s_i^*)\times n}\right)$\\
        $S_i^* \gets S_i^n$ $\left(S_i^*\in\mathbb{R}^{n\times s_i^*}\right)$ \\

    \KwOut{$W_i^*$, $P_{W_i^*}$, $S_i^*$ and $P_{S_i^*}$
    }
\end{algorithm}

\begin{algorithm}[ht] \label{alg:W_gi}
    \caption{Numerical calculation of the basis of $\W_{g,i}^*$ and the canonical projection $P_{W_{g,i}^*}$at node $i$}
    \LinesNumbered
    \KwIn{
    Insertion map $W_i^*$ and $S_i^*$.\\Canonical projection $P_{W_i^*}$}
    Find a basis $V_i$ such that
    \[\image{S_i^*} = \image{\matrices{W_i^*& V_i}} \]
    \[\dim{S_i^*} = \dim{W_i^*}+\dim{V_i}\]\\
    Find an arbitrary $L^{(\mathrm{a})}_i \in \R^{n\times p_i}$ such that 
            \[P_{W_{i}^*} \left(A+L^{(\mathrm{a})}_i C_i\right) W_{i}^* = \mathbf{0}_{(n-w_{i}^*)\times w_{i}^*}.\]
    
    $A_{L_i|V_i} \gets V_i^\top \left(A+L^{(\mathrm{a})}_i C_i\right) V_i$

    $\beta_i(\lambda) \gets$ the minimal polynomial of $A_{L_i|V_i}$

    Distinguish the ``good'' and ``bad'' zeros of $\beta_i(\lambda)$: $\beta_i(\lambda) = \beta_{g,i}(\lambda) * \beta_{b,i}(\lambda)$

    $\bar X_{b,i}^* \gets \Ker \beta_{b,i}(A_{L_i|V_i})$



    $W_{g,i}^* \gets \matrices{W_i^* & \bar X_{b,i}^*}$ $\left(W_{g,i}^*\in\mathbb{R}^{n\times w_{g,i}^*}\right)$\\
    $P_{W_{g,i}^*} \gets \pr{\kernel \left({W_{g,i}^*}^\top\right)}^\top$ $\left(P_{W_{g,i}^*}\in\mathbb{R}^{(n-w_{g,i}^*)\times n}\right)$\\
\KwOut{$A_{L_i|V_i}$, $W_{g,i}^*$ and $P_{W_{g,i}^*}$}
\end{algorithm}

\begin{algorithm}[ht] \label{alg:Li}
    \caption{Numerical calculation of the output injection map $L_i$ to obtain the desired spectrum of $A_{L_i}|\X/\W_{g,i}^*$ at node $i$}
    \LinesNumbered
    \KwIn{
    The set of symmetric desired (stable) poles $\bar \Lambda_{s,i}$ containing $s_i^*$ complex numbers.\\
    Insertion map $W_{g,i}^*$ and canonical projection $P_{W_{g,i}^*}$.}
            Find an arbitrary $L^{(\mathrm{a})}_i \in \R^{n\times p_i}$ such that 
                \[P_{W_{i}^*} \left(A+L^{(\mathrm{a})}_i C_i\right) W_{i}^* = \mathbf{0}_{(n-w_{i}^*)\times w_{i}^*}.\] \\
            $A_i^0 \gets P_{S_{i}^*}\left(A+L^{(\mathrm{a})}_i C_i\right) P_{S_{i}^*}^\dagger$ \\
            $S_i^Y \gets \text{ basis of } C_i\Ss_i^* \text{ in } \Y_i$ \\
            $P_i^Y \gets \pr{\kernel {S_i^Y}^\top}^\top$ \\
            $C_i^0 \gets P_i^Y C_i P_{S_{i}^*}^\dagger$ \\
            Choose $L_i^0\in\mathbb{R}^{(n-s_{i}^*)\times s_{i}^*}$ such that
                \[
                \kappa \pr{A_i^0 + L_i^0 C_i^0} = \bar \Lambda_{s,i}
                \]\\
            $L_i \gets L_i^{(\mathrm{a})} + P_{W_{g,i}^*}^\dagger L_i^0 P_i^Y$ \\
                \vspace{0.1in}
    \KwOut{$L_i$}
\end{algorithm}

\section{Distributed Unknown Input Observer Design}\label{sec:DUIO}

In this section, we propose the design of DUIO for both continuous- and discrete-time LTI systems \eqref{eq:systemDecomposition_ct} and \eqref{eq:systemDecomposition_dt} with measurements \eqref{eq:output}, building on the geometric subspace decomposition introduced earlier. At each node, convergence of the estimated state within the quotient space $\X/\W_{g,i}^*$ is ensured through an appropriate choice of $L_i$. The components of the estimated state associated with the subspace $\W_{g,i}^*$ are reconstructed via consensus instead. Since the consensus mechanisms differ between continuous- and discrete-time systems, the resulting geometric conditions also differ slightly in each case; their connections and distinctions are analyzed in the discussion that follows.

\subsection{DUIO Design for Continuous-time Systems}
The dynamics of the continuous-time DUIO for node $i$, $i\in\mathbf{N}$, are given by
 \begin{multline}\label{eq:DUIO_ct}
    \dot{\hat{x}}_i(t) = (A+L_iC_i)\hat{x}_i(t)-L_i y_i(t) + B_i u_i(t) \\
    + \chi_i W_{g,i}^* {W_{g,i}^*}^\top \sum_{j=1}^N a_{ij} (\hat x_j(t) - \hat x_i(t))\\
    + \gamma_i W_{g,i}^*\mathrm{sign}\left( {W_{g,i}^*}^\top \sum_{j=1}^N a_{ij} (\hat x_j(t) - \hat x_i(t)) \right),   
 \end{multline}
where $\chi_i,\ \gamma_i\in\R_{>0}$ are the coupling gains of the consensus terms to be designed. {The overall structure follows a standard consensus-based observer architecture augmented with a sliding-mode injection term, a commonly adopted mechanism for handling unknown inputs (see e.g., \cite{edwards1998sliding,cao2023bounded}). The key distinction lies in the insertion map $W_{g,i}^*\in\mathbb{R}^{n\times w_{g,i}^*}$ and the associated observer gain $L_i\in\mathbb{R}^{n\times p_i}$ which cannot be obtained through conventional observer constructions. Instead, they are explicitly derived from the proposed subspace decomposition developed in Section~\ref{sec:space_decomposition}.}

From \eqref{eq:error}, the error dynamics at node $i$ can be expressed as
    \begin{align}\label{eq:e:dynamics}
    \dot{e}_i(t)&=A_{L_i}e_i(t)+\bar{B}_i \bar{u}_i(t)\notag\\
        &-\chi_i W_{g,i}^*{W_{g,i}^*}^\top \sum_{j=1}^N a_{ij} (e_i(t) - e_j(t))\notag\\
        &-\gamma_i W_{g,i}^*\mathrm{sign}\left( {W_{g,i}^*}^\top \sum_{j=1}^N a_{ij} (e_i(t) - e_j(t)) \right) .
    \end{align}
%
Let $T_i \coloneq \matrices{W_{g,i}^*&P_{W_{g,i}^*}^\top}\in\mathbb{R}^{n\times n}$, $W_{g,i}^*\in\mathbb{R}^{n\times w_{g,i}^*}$, $P_{W_{g,i}^*}\in\R^{(n-w_{g,i}^*)\times n}$, be an orthonormal transformation adapted to the subspace decomposition at node $i$. Accordingly, the state space can be expressed as
$
    \X = \W_{g,i}^*\oplus\M_i^*, 
    $
where $\M_i^*\simeq\X/\W_{g,i}^*$ is an arbitrary complementary subspace to $\W_{g,i}^*$, and its orthonormal basis can be represented by $P_{W_{g,i}^*}^\top$. Hence, one can obtain
\begin{subequations}
\begin{align*}
   &{W_{g,i}^*}^\top P_{W_{g,i}^*}^\top = \mathbf{0}_{w_{g,i}^*\times(n-w_{g,i}^*)},\\
   &{W_{g,i}^*}^\top W_{g,i}^* = I_{w_{g,i}^*}, \ 
   P_{W_{g,i}^*} P_{W_{g,i}^*}^\top = I_{n-w_{g,i}^*}.
   \end{align*}
\end{subequations}
%
From the commutative diagram in Fig.~\ref{fig:Wg_decompose}, the following relationships hold
\begin{subequations}
\begin{align*}
    &W_{g,i}^*(A_{L_i}|\W_{g,i}^*) = A_{L_i}W_{g,i}^*,\\ 
    &(A_{L_i}|\mathscr{X}/\W_{g,i}^*)P_{W_{g,i}^*} = P_{W_{g,i}^*} A_{L_i}.
   \end{align*}
\end{subequations}
Consequently, under the transformation by $T_i$, the matrix representation of $A_{L_i}|\W_{g,i}^*$ and $A_{L_i}|\X/\W_{g,i}^*$ can be written as
\begin{subequations}\label{eq:AL}
\begin{align}
        &\Tilde{A}_{L_i}\coloneq \mathrm{Mat}(A_{L_i}|\W_{g,i}^*)=W_{g,i}^\top A_{L_i} W_{g,i}\label{eq:tildeAL}\\
        &\Bar{A}_{L_i}\coloneq \mathrm{Mat}(A_{L_i}|\X/\W_{g,i}^*)=P_{W_{g,i}^*} A_{L_i} P_{W_{g,i}^*}^\top.
\label{eq:barAL}
   \end{align}
\end{subequations}
where $\Bar{A}_{L_i}$ is Hurwitz as $\kappa(A_{L_i}|\X/\W_{g,i}^*)$ is assigned inside the stable region by $L_i$, as given in Section~\ref{sec:space_decomposition}. 
From \eqref{eq:global_error}, the dynamics of the overall estimation error can be expressed as
\begin{equation}
    \begin{aligned}
    \dot{e}(t)=&\left(A_L  - \boldsymbol{\chi} W_g^* {W_g^*}^\top(\Lap\otimes I_n)\right)e(t)\\
     &+ \bar{B}\bar{u}(t)- \boldsymbol{\gamma} W_g^* \mathrm{sign}\left({W_g^*}^\top(\Lap\otimes I_n)e(t)\right)
\end{aligned}\label{eq:error_ct}
\end{equation}
where
\begin{subequations}
\begin{align*}
    &A_L = \diag(A_{L_1},\cdots,A_{L_N}),\
    W_g^* = \diag(W_{g,1}^*,\cdots,W_{g,N}^*),\\
    &\bar{B} = \diag(\bar{B}_1,\cdots,\bar{B}_N),\ 
    \bar{u} = \col(\bar{u}_1,\cdots,\bar{u}_N),\\
    &\boldsymbol{\chi} = \diag(\chi_1,\cdots,\chi_N)\otimes I_n,\ 
    \boldsymbol{\gamma} = \diag(\gamma_1,\cdots,\gamma_N)\otimes I_n.
    \end{align*}
\end{subequations}
Define $M=\diag(P_{W_{g,1}^*}^\top,\cdots,P_{W_{g,N}^*}^\top)$, and construct the orthonormal transformation $T\in\mathbb{R}^{nN\times nN}$ as
\begin{equation}\label{eq:basis_T}
    T \coloneq \matrices{W_g^*&M}.
\end{equation}
Then, the matrix representation of the map $A_L-\boldsymbol{\chi} W_g^*{W_g^*}^\top(\Lap\otimes I_n)$, by the change of basis $T$, follows
\begin{align}\label{eq:Aa&Ab}
    T^\top\left( A_L-\boldsymbol{\chi} W_g^*{W_g^*}^\top(\Lap\otimes I_n)\right) T=\begin{bmatrix}
        A_a&\star\\
        \mathbf{0}&A_b
    \end{bmatrix}
\end{align}
where
\begin{subequations}
\begin{align*}
        &A_a = \Tilde{A}_L - \bar{\boldsymbol{\chi}} {W_g^*}^\top (\Lap\otimes I_n)W_g^*,
        A_b = \diag(\bar{A}_{L_1},\cdots,\bar{A}_{L_1}),\\
        &\Tilde{A}_L=\diag(\Tilde{A}_{L_1},\cdots,\Tilde{A}_{L_1}),\ 
        \bar{\boldsymbol{\chi}}=\diag(\chi_1 I_{s_1},\cdots,\chi_N I_{s_N}).
    \end{align*}
\end{subequations}
The representation of $\star$ is omitted as it does not affect the eigenvalues of the transformed matrix.
Let $\epsilon\coloneq T^\top e$ with $\epsilon\coloneq[\epsilon_a^\top\ \epsilon_b^\top]^\top$. Then, in the new coordinates, \eqref{eq:error_ct} can be expressed as
\begin{align}
        &\begin{bmatrix}
        \dot\epsilon_a(t)\\\dot\epsilon_b(t)
    \end{bmatrix}=\begin{bmatrix}
        A_a&\star\\
        \mathbf{0}&A_b
    \end{bmatrix}\begin{bmatrix}
        \epsilon_a(t)\\\epsilon_b(t)
    \end{bmatrix}+\begin{bmatrix}
        {W_g^*}^\top\bar{B}\\
        \mathbf{0}
    \end{bmatrix}\bar{u}(t)\notag\\
    &\quad -\begin{bmatrix}
        \bar{\boldsymbol{\gamma}}\mathrm{sign}\left({W_g^*}^\top(\Lap\otimes I_n)(W_g^*\epsilon_a(t)+M\epsilon_b(t)) \right)\\
        \mathbf{0}
    \end{bmatrix}\label{eq:epsilon}
\end{align}
where $\bar{\boldsymbol{\gamma}}=\diag(\gamma_1 I_{w^*_{g,1}},\cdots,\gamma_N I_{w^*_{g,N}})$.
From \eqref{eq:epsilon}, we observe that the overall error dynamics of the DUIO can be decomposed into two components: $\epsilon_a$, which must be stabilized through consensus terms associated with the subspace $\W_{g,i}^*$, and $\epsilon_b$, which is stabilized via local output injections corresponding to the quotient space  $\X/\W_{g,i}^*$. 
The following assumptions are required to proceed with the main results.
\begin{assumption}\label{as:ct}
The subspace $\W_{g,i}^*$ of each node $i$ has the following joint property:
    \begin{equation}
        \bigcap_{i=1}^N\W_{g,i}^* = 0. \label{con:con_for_ct_system}
    \end{equation}
\end{assumption}

\begin{lemma}\cite[Lemma~4]{kim2019completely}
    If Assumptions~\ref{as:connected} and \ref{as:ct} hold, ${W_g^*}^\top(\Lap\otimes I_n)W_g^*$ is positive definite.\label{lem:positivedefinite}
\end{lemma}

\begin{assumption}\label{as:input_bound}
    The unknown input signal $\bar{u}_i$, $i\in\mathbf{N}$, is bounded such that
     \begin{equation}
        \|\bar{u}_i\|_\infty \le \bar{u}_{\max},
     \end{equation}
where $\bar{u}_{\max}$ is a known constant.
\end{assumption}
The condition \eqref{con:con_for_ct_system} in Assumption~\ref{as:ct} is a joint detectability condition and, to the best of our knowledge, has not been previously reported in the literature. Comparing condition \eqref{con:con_for_ct_system} in Assumption~\ref{as:ct} with the rank condition \eqref{eq:rank_con}, we observe that \eqref{eq:rank_con} requires a unit relative degree for all $i\in \mathbf N$, implying complete reconstruction of the unknown input from each node's measurement. This is significantly more stringent than the condition presented in Assumption~\ref{as:ct}, which only requires the joint null space of the subspaces to be trivial. For example, if a node $i$ ($i\in\mathbf{N}$) has $\mathrm{rank}(\bar{B}_i)>\mathrm{rank}(C_i)$, it cannot fulfill the condition specified in \eqref{eq:rank_con}. In a specific scenario where a node experiences a failure resulting in no output (i.e., $C_i = 0$, leading to $\W_{g,i}^* = \X$), it inherently cannot satisfy \eqref{eq:rank_con} since $C_i\bar B_i =0$ holds for any $\bar B_i$. However, Assumption~\ref{as:ct} remains feasible with the cooperation of other nodes, provided $\cap_{k \in \mathbf{N}\backslash \{i\}}\W_{g,k}^* = 0$. This allows for the construction of an asymptotically convergent DUIO under our proposed design. As a special case, if there are no unknown inputs at any node ($\bar{B}_i = 0$ for all $i \in \mathbf{N}$), $\W^*_{g,i}$ becomes the undetectable subspace at node~$i$. In this case, Assumption~\ref{as:ct} coincides with the conventional joint detectability assumption commonly used in the literature (see \cite{yang2023plug}, \cite{wang2017distributed}).
Assumption~\ref{as:input_bound}, which concerns the boundedness of unknown inputs, is a standard condition in the control literature for handling unknown inputs \cite{edwards1998sliding}.

In the following, we show that the designed DUIO \eqref{eq:DUIO_ct} achieves the convergence property in \eqref{def:dUIO_ct}.

\begin{theorem}\label{thm:main_ct}
Consider the continuous-time system with distributed measurements given in \eqref{eq:con_sys} and \eqref{eq:output} and the DUIO in \eqref{eq:DUIO_ct}. Under Assumptions~\ref{as:connected}-\ref{as:input_bound}, if $L_i$ is chosen such that $\kappa(A_{L_i}|\X/\W_{g,i}^*)\subset\C_g,\ \forall i \in \mathbf{N}$, the estimation error $e(t)$ converges to zero if the coupling gains satisfy 
\begin{align}
        \chi_i&> \frac{\left\|\Tilde{A}_L\right\|_2}{\sigma_{\min}\left({W_g^*}^\top(\Lap\otimes I_n)W_g^*\right)},\notag\\ 
        \gamma_i&>\bar{u}_{\max}\max_{i\in\mathbf{N}}\left(\left\| \bar{B}_i \right\|_1\right)\max_{i\in\mathbf{N}}\left(\left\|W_{g,i}^*\right\|_{\infty}\right)
    \label{con:gain}
\end{align}
for all $i \in \mathbf{N}$.
\end{theorem}
\begin{proof}
As formulated in \eqref{eq:Aa&Ab}, $A_b$ is Hurwitz if $\kappa(A_{L_i}|\X/\W_{g,i}^*)\subset\C_g,\ \forall i \in \mathbf{N}$. Then, it is immediate to show 
$\underset{t\rightarrow\infty}{\lim}\epsilon_b(t)=0$ as $\dot\epsilon_b(t)=A_b\epsilon_b(t)$ (see \eqref{eq:epsilon}). Therefore, it suffices to prove the convergence of the following subsystem
\begin{equation}\label{eq:epsilon_a}
\begin{aligned}
    \dot\epsilon_a(t) = &A_a\epsilon_a(t) + {W_g^*}^\top \bar{B}\bar{u}(t)\\
    &-\bar{\boldsymbol{\gamma}}\mathrm{sign}\left({W_g^*}^\top(\Lap\otimes I_n)W_g^*\epsilon_a(t) \right)\,.
\end{aligned}
\end{equation}
For ease of notation, let us define $\boldsymbol{\Theta}\coloneq {W_g^*}^\top(\Lap\otimes I_n)W_g^*$, which is positive definite by Lemma~\ref{lem:positivedefinite}. Consider the following Lyapunov candidate
\begin{equation*}
    V(t) = \epsilon_a(t)^\top \boldsymbol{\Theta} \epsilon_a(t) .
\end{equation*}
The derivative of $V(t)$ along the trajectory of \eqref{eq:epsilon_a} follows
\begin{equation}
\begin{split}
        &\dot{V}(t)=2\epsilon_a(t)^\top\left(\boldsymbol{\Theta}A_a\right)\epsilon_a(t) + 2\epsilon_a(t)^\top \boldsymbol{\Theta} {W_g^*}^\top \bar{B} \bar{u}(t) \\
        &\quad\quad\quad - 2 \epsilon_a(t)^\top \boldsymbol{\Theta}\bar{\boldsymbol{\gamma}} \sign\left(\boldsymbol{\Theta}\epsilon_a(t) \right)\\
        &=2(\boldsymbol{\Theta}\epsilon_a(t))^\top\left(\Tilde{A}_L \boldsymbol{\Theta}^{-1}-\bar{\boldsymbol{\chi}}\right)(\boldsymbol{\Theta}\epsilon_a(t))\\
        &\quad+2\bar{u}(t)^\top\bar{B}^\top {W_g^*}(\boldsymbol{\Theta}\epsilon_a(t))
         - 2(\boldsymbol{\Theta}\epsilon_a(t))^\top \bar{\boldsymbol{\gamma}} \sign\left(\boldsymbol{\Theta}\epsilon_a(t) \right) .
\end{split}\label{eq:Vdot}
\end{equation}
Denote $\zeta \coloneq \boldsymbol{\Theta}\epsilon_a$, and substitute this into \eqref{eq:Vdot}. Then, it can be inferred that
\begin{equation*}
    \begin{split}
        &\dot{V}(t)\le 2\left(\left\|\Tilde{A}_L \boldsymbol{\Theta}^{-1}\right\|-\chi_{\min}\right)\|\zeta(t)\|_2^2 \\
        &\quad\quad\quad + 2\left\|\bar{u}(t)^\top \bar{B}^\top W_g^* \zeta(t) \right\|_{\infty} - 2 \gamma_{\min}\|\zeta(t)\|_1\\
        &\le 2\left(\left\|\Tilde{A}_L\right\| \left\|\boldsymbol{\Theta}^{-1}\right\|-\chi_{\min}\right)\|\zeta(t)\|_2^2 \\
        &\ + 2\left\|\bar{u}(t)^\top \right\|_{\infty} \left\|\bar{B}^\top \right\|_{\infty} \left\| W_g^* \right\|_{\infty} \left\| \zeta(t)\right\|_{\infty} - 2 \gamma_{\min}\|\zeta(t)\|_1\\
        &\le -2 \left(\chi_{\min} - \frac{\left\|\Tilde{A}_L\right\|}{\sigma_{\min}\left(\boldsymbol{\Theta}\right)}\right)\|\zeta(t)\|_2^2 \\
        &\ -\!2 \left( \gamma_{\min}\!-\! \bar{u}_{\max}\max_{i\in\mathbf{N}}\left(\left\| \bar{B}_i \right\|_1\right)\max_{i\in\mathbf{N}}\left(\left\|W_{g,i}^*\right\|_{\infty}\right)\right)\|\zeta(t)\|_1
    \end{split}
\end{equation*}
where $\chi_{\min} = \underset{i\in\mathbf{N}}{\min}(\chi_i)$, $\gamma_{\min} = \underset{i\in\mathbf{N}}{\min}(\gamma_i)$.

The conditions \eqref{con:gain} ensure that $\dot{V}(t)$ is negative definite. As such, the trajectory of $\epsilon(t)$ can achieve asymptotic convergence towards zero. As $e(t) = T\epsilon(t)$, where $T$ is based on the definition \eqref{eq:basis_T}, the convergence of $e(t)$ can be guaranteed, which completes the proof.
\end{proof}

\begin{remark}
    It is worth noting that the right-hand side of \eqref{eq:e:dynamics} is discontinuous due to the presence of the $\mathrm{sign}(\cdot)$ function. However, since $\mathrm{sign}(\cdot)$ is locally bounded and measurable, a solution exists in the sense of Filippov \cite{filippov1960differential}. Furthermore, the sliding mode term
    \begin{equation*}
        \gamma_i\bar{B}_i^\dagger W_{g,i}^* \sign\left( {W_{g,i}^*}^\top \sum_{j=1}^N a_{ij} (\hat x_j(t) - \hat x_i(t)) \right),\label{eq:input_reconstruct}
    \end{equation*}
    can be used to estimate the unknown input $\bar{u}_i$ by passing it through a low pass filter \cite[Chapter~1.2]{edwards1998sliding}.
\end{remark}
\begin{corollary}[Existence of the design]
    Under Assumptions~\ref{as:connected} and \ref{as:input_bound}, and the conditions in \eqref{con:gain}, the DUIO consisting of local observers in the form \eqref{eq:DUIO_ct} will generate estimation errors that converge to zero if and only if \eqref{con:con_for_ct_system} is satisfied.
\end{corollary}
\begin{proof}
    \textit{(Sufficiency)} -- If condition \eqref{con:con_for_ct_system} holds, then by Theorem~\ref{thm:main_ct}, we conclude that there always exists a DUIO design in the form of \eqref{eq:DUIO_ct} that ensures the estimation error $e$ asymptotically converges to zero.

    \textit{(Necessity)} -- To prove the necessity of condition \eqref{con:con_for_ct_system}, we proceed by contradiction and assume that there exists a nontrivial subspace $\mathscr V \subset \mathscr X$ such that
    \begin{equation*}
        \bigcap_{i=1}^N\W_{g,i}^* = \mathscr V \neq 0,
    \end{equation*}
    and the estimation error of the DUIO in the form of \eqref{eq:DUIO_ct} can decay exponentially.
    We exclude the case where $\W_{g,i}^* \cap  \kernel C_i = 0$ for any $i \in \mathbf{N}$, since in this scenario, a full-state observer could easily be implemented at that node (by left inverting $\matrices{P_{W_{g,i}}^\top & C_i^\top }^\top$), and the estimated state could then be transmitted to other nodes. Therefore, let us select a nonzero vector $v$ such that
    \begin{equation*}
        v \in \bigcap_{i=1}^N \pr{\W_{g,i}^* \cap \kernel C_i} \neq 0.
    \end{equation*}
    It follows that
$
        A_{L_i}v = \lambda v,\ \forall i \in \mathbf{N},\,
$
    where $\lambda$ is an unstable scalar. Let $\bar e = \mathbf{1}_N \otimes v$. Since $\bar e \in \kernel (\Lap \otimes I_n)$ and $A_L \bar e = \lambda \bar e$, selecting $e_0=\bar{e}$ as the initial condition for the estimation error $e(t)$ that evolves according to \eqref{eq:error_ct} yields
    \begin{equation*}
        \dot{e}(t)=\lambda e(t)  + \bar{B}\bar{u}(t).
    \end{equation*}
    This contradicts the asymptotic stability hypothesis. Hence, condition \eqref{con:con_for_ct_system} must be satisfied.
\end{proof}

\subsection{DUIO Design for Discrete-time Systems}\label{sec:discrete_DUIO}
The DUIO for discrete-time systems \eqref{eq:systemDecomposition_dt} \eqref{eq:output} is investigated in this subsection. Unlike the continuous-time setting, where the local estimators can achieve consensus under continuous-time (infinite) communication, the discrete-time version reaches consensus after a finite number of communication exchanges at each time step.
Let $d$ denote the number of communications among neighboring nodes within a single system time interval.
We now present the discrete-time DUIO design.
{\allowdisplaybreaks
\begin{subequations}\label{eq:dis_DUIO}
\begin{align}
    &z_i(t\!+\!1) = \bar A_{L_i}  z_i(t) - P_{W_{g,i}^*}L_i y_i(t) + P_{W_{g,i}^*}B_i u_i(t),\label{eq:zi}\\
    &\zeta_i(0,t) = E_iz_i(t) + F_iy_i(t),\label{eq:zeta0}\\
    &\zeta_i(k,t) = \varpi_{ii}\zeta_i(k-1,t) + \sum_{j=1,j\neq i}^N\varpi_{ij}\zeta_j(k-1,t),\label{eq:zeta_k}\\
    &\hat{x}_i(t) = N\zeta_i(d,t),\label{eq:hat_xi}
\end{align}
\end{subequations}
}
where $z_i\in\mathbb{R}^{n-w_{g,i}^*}$, $\bar A_{L_i}\in\mathbb{R}^{(n-w_{g,i}^*)\times (n-w_{g,i}^*)}$ is defined in \eqref{eq:barAL}, $P_{W_{g,i}^*}\in\R^{(n-w_{g,i}^*)\times n}$ and $L_i\in\R^{n\times p_i}$ have been designed in Section~\ref{sec:space_decomposition}, $E_i\in\R^{n\times(n-w_{g,i}^*)}$ and $F_i\in\R^{n\times p_i}$ can be obtained from \eqref{eq:compute_Ei_Fi}, $k\in\mathbf{d}=\{1,2,\cdots,d\}$, and the consensus matrix $\mathbf{W}=[\varpi_{ij}]\in\R^{N\times N}$ is designed as 
\begin{equation}
    \mathbf{W} = I_N - \frac{1}{\mu}\mathcal{L}\label{eq:W}
\end{equation}
with $\mu>\frac{\lambda_N(\mathcal{L})}{2}$.

The discrete-time DUIO design in \eqref{eq:dis_DUIO} is established through the following Lemmas and a main Theorem. First, Lemma~\ref{lemma:z_i} ensures that the reduced-order local estimator $z_i$ in \eqref{eq:zi} asymptotically converges to $P_{W_{g,i}^*}x$. Next, Lemma~\ref{lemma:mu} identifies the optimal choice of $\mu$ (as a function of $\mathcal{L}$) for the consensus matrix $\mathbf{W}$ defined in \eqref{eq:W}, enabling the consensus defined in \eqref{eq:zeta0} and \eqref{eq:zeta_k} at the fastest convergence rate. Finally, Theorem~\ref{theorem:dis_DUIO} establishes that the overall estimation error is bounded under the proposed design in \eqref{eq:dis_DUIO}.
\begin{lemma}
    For each local estimator designed in \eqref{eq:zi}, the estimated state $z_i$ asymptotically converges to $P_{W_{g,i}^*}x$.\label{lemma:z_i}
\end{lemma}
\begin{proof}
    Let $\varrho_i\coloneq P_{W_{g,i}^*} x - z_i$, we have that
    \begin{align*}
        \varrho_i(t+1) &= P_{W_{g,i}} A x(t) + P_{W_{g,i}} \bar B_i \bar u_i(t) \\&\quad - \bar A_{L_i}z_i(t) + P_{W_{g,i}}L_iC_ix(t)\\
         = P_{W_{g,i}}& A_{L_i} x(t) + P_{W_{g,i}} \bar B_i \bar u_i(t) - \bar A_{L_i}z_i(t).
    \end{align*}
It is clear that $\kernel P_{W_{g,i}^*} = \W_{g,i}^*$, since the map $P_{W_{g,i}^*}:\X\rightarrow\X/\W_{g,i}^*$ is the canonical projection. By definition, we have $\image\bar{B}_i\subseteq\W_{g,i}^*$, which leads to $P_{W_{g,i}^*}\bar{B}_i=0$. Moreover, according to the commutative diagram in Fig.~\ref{fig:Wg_decompose}, one gets $P_{W_{g,i}^*} A_{L_i}=\bar A_{L_i} P_{W_{g,i}^*}$. As a consequence, we obtain
    \begin{equation*}
        \varrho_i(t+1) = \bar{A}_{L_i} P_{W_{g,i}^*} x(t) - \bar A_{L_i}z_i(t) = \bar{A}_{L_i}\varrho_i(t).
    \end{equation*}
Since $\bar{A}_{L_i}$ is Schur by appropriate design of $L_i$ in Section~\ref{sec:space_decomposition}, it follows immediately that
$
        \lim_{t\rightarrow\infty}\left\|\varrho_i(t)\right\|=0,
$
    which completes the proof.
\end{proof}

\begin{assumption}\label{asm:discrete}
The canonical projections $P_{W_{g,i}^*}:\X\rightarrow\X/\W_{g,i}^*$ of each node $i$ have the following joint property
\begin{equation}
        \bigcap_{i=1}^N \kernel \matrices{P_{W_{g,i}^*} \\ C_i} = 0.\label{con:con_for_dt_system}
\end{equation}
\end{assumption}
This assumption is similar to Assumption~\ref{as:ct}, however, the discrete-time joint geometric condition \eqref{con:con_for_dt_system} additionally accounts for the information mapped from $\X$ onto $\image{C_i}$. This aspect is neglected in the continuous-time case, since the subspace $\left(\W_{g,i}^*\cap\kernel{C_i}\right)$ is not necessarily $A_{L_i}$-invariant when nontrivial. By contrast, in the discrete-time framework, this component is incorporated because consensus is achieved algebraically. Consequently, the discrete-time geometric condition \eqref{con:con_for_dt_system} is less conservative than the continuous-time geometric condition \eqref{con:con_for_ct_system} as $\kernel{P_{W_{g,i}^*}}=\W_{g,i}^*$.

Under Assumption~\ref{asm:discrete}, it follows from \cite[Chapter~0.12]{wonham1985linear} that
\begin{equation*}
    \sum_{i=1}^N \left( \image P_{W_{g,i}^*}^\top +\image C_i^\top \right) = \X
\end{equation*}
which implies the existence of $E_i$ and $F_i$, $i\in\mathbf{N}$, such that
\begin{equation}
  \sum_{i=1}^{N}(E_iP_{W_{g,i}^*} + F_iC_i) = I_n.\label{eq:Ei_Fi}
\end{equation}
It should be noted that the solutions of $E_i$ and $F_i$ are not unique when there exist $i$, $j\in\mathbf{N}$ ($i\neq j$) such that $\kernel{\matrices{P_{W_{g,i}^*}\\ C_i}}\cap\kernel{\matrices{P_{W_{g,j}^*}\\ C_j}}\neq 0$. In this case, we compute them using the pseudoinverse. Let $R_i\coloneq \matrices{P_{W_{g,i}^*}\\C_i}$ and $R\coloneq\col_{i\in\mathbf{N}}(R_i)$, then $E_i$ and $F_i$ can be obtained as follows
\begin{multline}\label{eq:compute_Ei_Fi}
        \begin{bNiceMatrix}
            E_1&F_1&E_2&F_2&\cdots&E_N&F_N
            \CodeAfter 
              \tikz \draw [dashed,shorten > = 2pt, shorten < = 2pt] (1-|3) -- (last-|3) (1-|5) -- (last-|5) (1-|6) -- (last-|6) ;
        \end{bNiceMatrix} \\
        =R^\dagger=\left(R^\top R\right)^{-1} R^\top \hfill\\
        =\underset{\coloneq\Phi^{-1}}{\underbrace{\left(\sum_{i=1}^N R_i^\top R_i\right)^{-1}}}\matrices{R_1^\top&R_2^\top&\cdots&R_N^\top}\hfill\\
        =\begin{bNiceMatrix}
            \Phi^{-1}R_1^\top&\Phi^{-1}R_2^\top&\cdots&\Phi^{-1}R_N^\top
            \CodeAfter 
              \tikz \draw [dashed,shorten > = 2pt, shorten < = 2pt] (1-|2) -- (last-|2) (1-|3) -- (last-|3) (1-|4) -- (last-|4) ;
        \end{bNiceMatrix} \hfill\\
        ={\setlength{\arraycolsep}{4pt}
        \begin{bmatrix}
        \Phi^{-1}P_{W_{g,1}^*}^\top&\Phi^{-1}C_1^\top&\cdots&\Phi^{-1}P_{W_{g,N}^*}^\top&\Phi^{-1}C_N^\top\end{bmatrix}.}\hfill
\end{multline}

According to Assumption~\ref{as:connected} on connectivity, the eigenvalues of $\mathcal{L}$ corresponding to the communication graph $\mathcal{G}=\{\mathbf{N},\mathcal{E},\mathcal{A}\}$ comply with\cite{olfati2007consensus} 
    \begin{equation*}
        0 = \lambda_1(\mathcal{L}) < \lambda_2(\mathcal{L}) \leq \cdots \leq \lambda_N(\mathcal{L}).
    \end{equation*}
Note that all the eigenvectors of $\mathcal{L}$ are also the eigenvectors of $\mathbf{W}$ due to the construction of $\mathbf{W}$ in \eqref{eq:W}, and the corresponding eigenvalues are related by
\begin{align}
    \lambda_i(\mathbf{W}) = 1 - \frac{1}{\mu}\lambda_i(\mathcal{L}),\ i\in\mathbf{N}.\label{eq:eigval_W}
\end{align}
Therefore, the eigenvalues of $\mathbf{W}$ can be ordered in descending order as
$
        1 = \lambda_1(\mathbf{W}) > \lambda_2(\mathbf{W}) \geq \cdots \geq \lambda_N(\mathbf{W})>-1,
$
where the left and right eigenvectors corresponding to $\lambda_1(\mathbf{W}) = 1$ are given by $\frac{\mathbf{1}_{N\times 1}}{\sqrt{N}}$ as stated in Lemma~\ref{lemma:eig_Laplacian_1}. Consequently, $\mathbf{W}^k$ can be expressed as
    \begin{align*}
        \mathbf{W}^k = \frac{1}{N}\mathbf{1}_{N\times N}+\sum_{\tau=2}^N \left(\lambda_\tau(\mathbf{W})\right)^k v_\tau v_\tau^\top,
    \end{align*}
where $v_\tau$ is the eigenvector corresponding to $\lambda_\tau(\mathbf{W})$.
\begin{definition}[Q-linear convergence] \cite{wright2006numerical}
  For a sequence $\{\mathbf{W}^k\}$ that converges to $\mathbf{W}^*$, the convergence is Q-linear if there exists a constant $r\in(0,1)$ such that 
    \begin{equation*}
         \lim_{k\rightarrow\infty}\frac{\left\|\mathbf{W}^{k+1}-\mathbf{W}^*\right\|}{\left\|\mathbf{W}^{k}-\mathbf{W}^*\right\|^q} = r 
    \end{equation*} 
where $q$ denotes the order of convergence.
\end{definition} \label{def:Q_linear}
Clearly, a smaller value of $r$ leads to faster convergence. We now introduce the following lemma, which characterizes the properties of the convergence rate.
    
\begin{lemma}\label{lemma:mu}
    For a matrix $\mathbf{W}$ defined in \eqref{eq:W} over any non-complete graph $\mathcal{G}$ under Assumption~\ref{as:connected}, $\{\mathbf{W}^k\}$ converges Q-linearly with order $1$ to $\frac{1}{N}\mathbf{1}_{N\times N}$, and with a fastest rate
    \begin{equation}\label{eq:ConverRate}
        r=\frac{\lambda_N(\mathcal{L})-\lambda_2(\mathcal{L})}{\lambda_2(\mathcal{L})+\lambda_N(\mathcal{L})}
    \end{equation}
    if and only if
    \begin{equation}
    \mu = \frac{\lambda_2(\mathcal{L})+\lambda_N(\mathcal{L})}{2}.\label{eq:mu}
\end{equation}
\end{lemma}
\begin{proof}
Based on Definition~\ref{def:Q_linear}, the rate of convergence of the sequence $\{\mathbf{W}^k\}$ is 
    \begin{equation*}
         r= \lim_{k\rightarrow\infty}\frac{\left\|\mathbf{W}^{k+1}-\mathbf{1}_{N\times N}\right\|}{\left\|\mathbf{W}^{k}-\mathbf{1}_{N\times N}\right\|} = \max(|\lambda_2(\mathbf{W})|,|\lambda_N(\mathbf{W})|).
    \end{equation*}
Given a non-complete graph $\mathcal{G}$, we have $\lambda_2(\mathcal{L})< \lambda_N(\mathcal{L})$. Let $\mu= \frac{\lambda_2(\mathcal{L})+\lambda_N(\mathcal{L})}{2}$, according to \eqref{eq:eigval_W}, one can obtain
    \begin{equation*}
        \lambda_2(\mathbf{W})=\frac{\lambda_N(\mathcal{L})-\lambda_2(\mathcal{L})}{\lambda_2(\mathcal{L})+\lambda_N(\mathcal{L})} = -\lambda_N(\mathbf{W}),
    \end{equation*}
    and therefore
    \begin{align*}
        \lim_{k\rightarrow\infty}\frac{\left\|\mathbf{W}^{k+1}-\frac{1}{N}\mathbf{1}_{N\times N}\right\|}{\left\|\mathbf{W}^{k}-\frac{1}{N}\mathbf{1}_{N\times N}\right\|} = 
        \frac{\lambda_N(\mathcal{L})-\lambda_2(\mathcal{L})}{\lambda_2(\mathcal{L})+\lambda_N(\mathcal{L})}.
    \end{align*}
Consider another constant $\mu'$ and the resulting $\mathbf{W}'$ by \eqref{eq:W} , such that
$\mu'>\frac{\lambda_2(\mathcal{L})+\lambda_N(\mathcal{L})}{2}$, one gets
    \begin{equation*}
        \lambda_2(\mathbf{W}')>\frac{\lambda_N(\mathcal{L})-\lambda_2(\mathcal{L})}{\lambda_2(\mathcal{L})+\lambda_N(\mathcal{L})},
    \end{equation*}
    and 
    \begin{equation}\label{eq:xx1}
         \lim_{k\rightarrow\infty}\frac{\left\|\mathbf{W}'^{k+1}-\frac{1}{N}\mathbf{1}_{N\times N}\right\|}{\left\|\mathbf{W}'^{k}-\frac{1}{N}\mathbf{1}_{N\times N}\right\|} > 
        \frac{\lambda_N(\mathcal{L})-\lambda_2(\mathcal{L})}{\lambda_2(\mathcal{L})+\lambda_N(\mathcal{L})}.
    \end{equation}
    Similarly, given $\mu''$ and $\mathbf{W}''$, such that $\frac{\lambda_N(\mathcal{L})}{2}<\mu''<\frac{\lambda_2(\mathcal{L})+\lambda_N(\mathcal{L})}{2}$, it can be inferred that 
    \begin{align*}
         \lambda_N(\mathbf{W}'')<\frac{\lambda_2(\mathcal{L})-\lambda_N(\mathcal{L})}{\lambda_2(\mathcal{L})+\lambda_N(\mathcal{L})},
    \end{align*}
and 
    \begin{align}\label{eq:xx2}
         \lim_{k\rightarrow\infty}\frac{\left\|\mathbf{W}''^{k+1}\!-\!\frac{1}{N}\mathbf{1}_{N\times N}\right\|}{\left\|\mathbf{W}''^{k}-\frac{1}{N}\mathbf{1}_{N\times N}\right\|}\! =\! -\lambda_N(\mathbf{W}'')\! >\!
        \frac{\lambda_N(\mathcal{L})\!-\!\lambda_2(\mathcal{L})}{\lambda_2(\mathcal{L})\!+\!\lambda_N(\mathcal{L})}.
    \end{align}
In view of \eqref{eq:xx1} and \eqref{eq:xx2}, it can be concluded that for the construction of $\mathbf{W}$ in \eqref{eq:W}, the faster Q-linear convergence rate \eqref{eq:ConverRate} is achieved provided \eqref{eq:mu}. This completes the proof.
\end{proof}
\begin{remark}
    In case $\mathcal{G}$ is a complete graph, it holds that $\lambda_2(\mathcal{L}) = \lambda_3(\mathcal{L})=\cdots = \lambda_N(\mathcal{L})=N-1$ \cite{cvetkovic2009introduction}.
     Let $\mu = \frac{\lambda_2(\mathcal{L})+\lambda_N(\mathcal{L})}{2}$ as in \eqref{eq:mu}. From the construction of $\mathbf{W}$ in \eqref{eq:W}, it follows that $\lambda_2(\mathbf{W}) = \lambda_3(\mathbf{W})=\cdots = \lambda_N(\mathbf{W})=0$ and $\mathbf{W} = \frac{1}{N}\mathbf{1}_{N\times N}$. This guarantees that the DUIO \eqref{eq:dis_DUIO} reaches consensus in a single step, as will be shown in Theorem 3.
\end{remark}

\begin{theorem}\label{theorem:dis_DUIO}
    Consider the discrete-time DUIO designed in \eqref{eq:dis_DUIO}, $E_i$ and $F_i$ calculated by \eqref{eq:compute_Ei_Fi}, and the consensus matrix $\mathbf{W}$ constructed by \eqref{eq:W} and \eqref{eq:mu}. Under the Assumptions~\ref{as:connected} and \ref{asm:discrete}, there exists a class $\mathcal{KL}$ function $\alpha(\cdot)$ such that the estimation error $e(t)$ satisfies
    \begin{equation}\label{eq:error_dt}
        \lim_{t\rightarrow +\infty}\|e(t)\|\le \alpha\left(\|x(t)\|,d\right).
    \end{equation}
\end{theorem}
\begin{proof}
    Let $\zeta(k,t)\coloneq\matrices{\zeta_1(k,t)^\top \,\, \zeta_2(k,t)^\top \,\,\cdots \,\,\zeta_N(k,t)^\top}^\top$. According to \eqref{eq:zeta_k}, we have
    \begin{equation*}
        \zeta(k,t) = \left(\mathbf{W}\otimes I_n\right) \zeta(k-1,t),
    \end{equation*}
   from which it is immediate to obtain
    \begin{equation*}
        \zeta(d,t)= \left(\mathbf{W}^d\otimes I_n\right)\zeta(0,t).
    \end{equation*}
    The amalgamated state estimation error across all nodes can be expressed as follows
    \begin{align*}
        e(t) &= \mathbf{1}_N \otimes x(t) - \underset{i\in\mathbf{N}}{\col}\left(\hat{x}_i(t)\right)\\
        &= \mathbf{1}_N \otimes x(t) - N  \left(\mathbf{W}^d\otimes I_n\right)\zeta(0,t)\\
       & =\mathbf{1}_N \otimes x(t) - \left(\mathbf{1}_{N\times N}\otimes I_n\right)\zeta(0,t) \\
        &\quad\ - \left(\left(\sum_{\tau=2}^N \left(\lambda_\tau(\mathbf{W})\right)^d v_\tau v_\tau^\top\right)\otimes I_n\right)\zeta(0,t)\\
        &=\mathbf{1}_N \otimes \left(x(t) - \sum_{i=1}^N \left( E_iz_i(t) + F_iy_i(t) \right)\right)\\
        &\quad\ - \left(\left(\sum_{\tau=2}^N \left(\lambda_\tau(\mathbf{W})\right)^d v_\tau v_\tau^\top\right)\otimes I_n\right)\zeta(0,t).
    \end{align*}
According to Lemma~\ref{lemma:z_i}, for each node $i$, we have that $\underset{t\rightarrow \infty}{\lim} \left\|P_{W_{g,i}^*} x(t) - z_i(t)\right\| = 0$, which leads to
    \begin{multline} 
        \lim_{t\rightarrow\infty}\left\|x(t) - \sum_{i=1}^N \left( E_iz_i(t) + F_iy_i(t) \right)\right\|\\
        =\lim_{t\rightarrow\infty}\left\|\sum_{i=1}^{N}E_i\left(P_{W_{g,i}^*}x(t)-z_i(t)\right)\right\|=0.
    \label{eq:error_of_rho}
    \end{multline}
Let $K_i \coloneq E_iP_{W_{g,i}^*} + F_iC_i$, one can obtain    
    \begin{multline*}
            \lim_{t\rightarrow\infty}\|e(t)\|
            =\left\|\left(\left(\sum_{\tau=2}^N \left(\lambda_\tau(\mathbf{W})\right)^d v_\tau v_\tau^\top\right)\otimes I_n\right)\zeta(0,t)\right\|\\
            \le\left\|\sum_{\tau=2}^N \left(\lambda_\tau(\mathbf{W})\right)^d v_\tau v_\tau^\top\right\|
            \sqrt{\sum_{i=1}^N\left\|  E_iz_i(t) + F_iy_i(t)\right\|^2}
    \end{multline*}
which further leads to  
    \begin{align*}
            \lim_{t\rightarrow\infty}\|e(t)\|
            &\le (N-1)\left(\lambda_2(\mathbf{W})\right)^d \max_{\tau\in\{2,\cdots,N\}}\left(\left\|v_\tau v_\tau^\top\right\|\right)\\
            &\myquad[10]\ \times\sqrt{\sum_{i=1}^N\left\|  K_ix(t)\right\|^2} \\
            &\hspace{-4mm} \le (N-1)\left(\lambda_2(\mathbf{W})\right)^d \sqrt{\lambda_{\max}\left(\sum_{i=1}^N K_i^\top K_i\right)}\|x(t)\|.
    \end{align*}
From \eqref{eq:Ei_Fi}, we have that $\sum_{i=1}^N K_i = I_n$. By applying Weyl inequality \cite{Horn_Johnson_2012}, one can obtain
       $\sum_{i=1}^N\lambda_{\max}(K_i) \leq 1$.
   According to \eqref{eq:compute_Ei_Fi}, we can deduce that $K_i$ is a positive semidefinite matrix. Hence, $0\leq \lambda_{\max}(K_i) \leq 1$, and consequently 
$  
       0\leq \lambda_{\max}\left(K_i^\top K_i\right)\leq 1.
$
Since
\begin{align*}
        &\lambda_{\max}\left(\sum_{i=1}^N K_i^\top K_i\right)=\max_{\eta\ne 0}\frac{\eta^\top \left(\sum_{i=1}^N K_i^\top K_i\right)\eta}{\eta^\top \eta}\\
        &\quad=\max_{\eta\ne 0}\left(\sum_{i=1}^N\frac{\eta^\top \left( K_i^\top K_i\right)\eta}{\eta^\top \eta}\right)
        \leq \sum_{i=1}^N\lambda_{\max}\left(K_i^\top K_i\right)\leq N,
\end{align*}
    we conclude that
    \begin{align}\label{eq:state_dependent_error}
        &\lim_{t\rightarrow\infty}\|e(t)\|
        \leq (N-1)\sqrt{N}\left(\lambda_2(\mathbf{W})\right)^d \|x(t)\|.
    \end{align}
This completes the proof of the estimation error bound $\alpha\left(\|x(t)\|,d\right)$ (see \eqref{def:dUIO_dt}).
\end{proof}
{Equation \eqref{eq:state_dependent_error} shows that the error bound is state dependent, scaling with $\|x(t)\|$ rather than being an absolute state-independent bound. In particular, if the system state $x(t)$ remains bounded, then the estimation error also remains bounded, with its magnitude scaling proportionally with the state magnitude. Since $0\le \lambda_2(\mathbf{W})<1$, the factor $(\lambda_2(\mathbf{W}))^d$ decays exponentially with the number of communication rounds $d$, so increasing $d$ improves the estimation accuracy. If the state is unbounded, the estimation error may also grow proportionally; however, this growth is attenuated by $(\lambda_2(\mathbf{W}))^d$, and thus larger $d$ improves the tracking accuracy of the discrete-time DUIO.}
\begin{remark}
    With the estimated state $\hat{x}_i(t)$, we can estimate the unknown input at the last time step, i.e., $\bar{u}_i(t-1)$, by
\begin{align*}
    \widehat{\bar{u}_i}(t-1) = \bar{B}_i^\dagger \left(\hat{x}_i(t) - A\hat{x}_i(t-1) - \bar{B}_iu_i(t-1)\right),
\end{align*}
assuming that the columns of $\bar{B}_i$ are linearly independent.
\end{remark}

\begin{corollary}[Existence of the design]
    Under Assumption~\ref{as:connected}, and with the consensus matrix defined in \eqref{eq:W} and \eqref{eq:mu}, the DUIO, where each local observer is given by  \eqref{eq:dis_DUIO}, produces an estimation error $ e(t) $ that is ultimately bounded by a class $\mathcal{KL}$ function as $t\rightarrow \infty$ if and only if condition \eqref{con:con_for_dt_system} holds.
\end{corollary}
\begin{proof}
\textit{(Sufficiency)} -- If condition \eqref{con:con_for_dt_system} is satisfied, then according to Theorem~\ref{theorem:dis_DUIO}, there always exists a DUIO design in the form of \eqref{eq:dis_DUIO} that guarantees the estimation error $e(t)$ is ultimately bounded by a class $\mathcal{KL}$ function $\alpha\left(\|x(t)\|,d\right)$.

    \textit{(Necessity)} -- To prove the necessity of condition \eqref{con:con_for_dt_system}, we proceed by contradiction and assume that there exists a nontrivial subspace $\mathscr J \subset \mathscr X$ such that 
    \begin{equation}
        \bigcap_{i=1}^N \kernel \matrices{P_{W_{g,i}^*} \\ C_i} = \mathscr J \neq 0,\label{eq:corollary2}
    \end{equation}
    and the estimation error of the DUIO in the form of \eqref{eq:dis_DUIO} will ultimately be bounded by a class $\mathcal{KL}$ function as $t\rightarrow \infty$.
    
Following \eqref{eq:corollary2}, there no longer exist $E_i$ and $F_i$, $i\in\mathbf{N}$, that satisfy \eqref{eq:Ei_Fi}.
Consequently, for any given $E_i$ and $F_i$, $i\in\mathbf{N}$, there exists a non-zero matrix $\Delta\in\mathbb{R}^{n\times n}$, such that
   \begin{equation*}
         \sum_{i=1}^{N}(E_iP_{W_{g,i}^*} + F_iC_i) = I_n + \Delta,
   \end{equation*}
   which, instead of following \eqref{eq:error_of_rho}, leads to
   \begin{equation*}
   \begin{split}
       \lim_{t\rightarrow\infty}\left\|x(t) - \sum_{i=1}^N \left( E_iz_i(t) + F_iy_i(t) \right)\right\|
       =\lim_{t\rightarrow\infty}\left\|\Delta x(t)\right\|.
   \end{split}
   \end{equation*}
   Then, the steady state error is bounded by   
   \begin{equation*}
       \lim_{t\rightarrow\infty}\|e(t)\|
       \leq (N-1)\sqrt{N}\left(\lambda_2(\mathbf{W})\right)^d \|x(t)\| + \left\|\Delta x(t)\right\|.
   \end{equation*}
    This result contradicts the assumption that the estimation error $e(t)$ is ultimately bounded by a class $\mathcal{KL}$ function. Therefore, condition \eqref{con:con_for_dt_system} is necessary.
\end{proof}
\smallskip
To conclude the presentation of the DUIO algorithms, Table~\ref{tab:comparison_ct_dt} provides a comparison between the proposed continuous-time and discrete-time approaches.
\begin{table}[ht]
    \centering
    \caption{A comparison between the continuous-time DUIO and discrete-time DUIO}
    \resizebox{0.49\textwidth}{!}{
    \begin{tabular}{c|c|c}
                        \hline
        \diagbox[dir=NW]{Properties}{DUIO Types} & Continuous-time  & Discrete-time \\
                         \hline
        Communication Graph&\text{Connected}&\text{Connected}\\
        Communication Type &\text{Continuous}&\text{Discrete}\\
        Time Scale         &\text{One}&\text{Two}\\
        Consensus Matrix & $\mathcal{L}$        & $\mathbf{W}=I_N - \frac{1}{\mu}\mathcal{L}$\\
        Geometric Condition&$\overset{N}{\underset{i=1}{\bigcap}}\W_{g,i}^* = 0$&$ \overset{N}{\underset{i=1}{\bigcap}} \kernel \matrices{P_{W_{g,i}^*} \\ C_i} = 0$\\
        Unknown Input Condition&$\|\bar{u}_i\|_\infty \le \bar{u}_{\max}$&\text{None}\\
        Convergence Type & $\underset{t\rightarrow\infty}{\lim} \|e(t)\| = 0$&$\underset{t\rightarrow\infty}{\lim} \|e(t)\| \leq \alpha\left(\|x(t)\|,d\right)$\\
        \hline
    \end{tabular}
    }
    \label{tab:comparison_ct_dt}
\end{table}

\begin{remark}
    Assumptions~\ref{as:ct} and~\ref{asm:discrete} are joint conditions across the distributed network, unlike in centralized settings, since $\bar B_i$ is unavailable to Node~$i$ but may lie in the column space of another node's input channel $B_j$. For a detailed discussion of centralized UIOs and their relationship to DUIOs from the geometric perspective, the reader is referred to \cite{zhao2025bridging}. 
    In contrast to existing works, which impose an additional joint detectability assumption on $A$ alongside \eqref{eq:rank_con} on ($\bar B_i$, $C_i$) (see, e.g., \cite[Equation~10]{yang2022state}), Assumptions~\ref{as:ct} and~\ref{asm:discrete} unify both into a single, more relaxed condition involving ($A$, $\bar B_i$, $C_i$).   
\end{remark}

\section{Simulation Results}\label{sec:simulation}
In this section, the effectiveness of the proposed DUIO schemes is validated in both continuous- and discrete-time settings, including a comparison based on their respective design criteria. Additionally, the advantages of the geometry-based subspace decomposition over existing methods are highlighted, demonstrating the reduced conservativeness achieved through $\W_{g,i}^*$. Finally, a case study on a DC microgrid system illustrates the practical applicability of the proposed DUIO  approach for real-world interconnected systems. 

\subsection{Numerical example}\label{sec:numerical_sim}
\subsubsection{Continuous-time system}
Consider a continuous-time system in the form of \eqref{eq:con_sys}, measured by $4$ sensor nodes. The system matrix $A$ and the input matrix $B$ are given by
\begin{align*}
    A &= \begin{bmatrix}
        0& 3&  0&  0&  0&  0\\
       -2& 0&  1&  0&  0&  0\\
       0&  0&  0&  2&  0&  0\\
       0&  0& -3&  -2& 0&  0\\
       0&  0&  0&  1&  0& -3\\
       0& 1.5& 0&  0&  4& 0
    \end{bmatrix},\\ 
    B^\top &=
    \begin{bmatrix}
     0&     1&     0&     0&     0&     1\\
     0&     0&     0&     1&     0&     0\\
     0&     0&     1&     0&     0&     1\\
    \end{bmatrix} = \begin{bmatrix}
    B_a^\top\\B_b^\top\\B_c^\top
    \end{bmatrix}.
\end{align*}
%
The input signal of the system follows $u(t)=\matrices{\mathtt{sin}(t)&2\mathtt{cos}(t)&2\mathtt{sin}(0.5t)}^\top$, which is decomposed as $u=\matrices{u_a&u_b&u_c}^\top$ in accordance with $B$. The known inputs at each node are defined as $u_1 = \matrices{u_a&u_b}^\top$, $u_2=\matrices{u_a&u_c}^\top$, $u_3=\matrices{u_b&u_c}^\top$ and $u_4=u_a$. Accordingly, the known and unknown input channels for each node can be characterized as follows
\begin{equation*}
\begin{aligned}
    &B_1 = \matrices{B_a&B_b},\ \bar B_1=B_c,\ 
    B_2 = \matrices{B_a&B_c},\ \bar B_2=B_b,\\
    &B_3 = \matrices{B_b&B_c}=\bar{B}_4,\ \bar B_3=B_a=B_4.
\end{aligned}
\end{equation*}
Finally, the output matrices of each node are defined as
\begin{equation*}
    \begin{aligned}
        &C_1=\begin{bmatrix}
        1& 0& 0& 0& 0& 0\\0& 1& 0& 0& 0& 0
        \end{bmatrix},
        C_3=\begin{bmatrix}
        0& 0& 1& 0& 0& 0\\0& 0& 0& 0& 1& 0
        \end{bmatrix},\\
        &C_2=\begin{bmatrix}
        0& 1& 0& 0& 1& 0
        \end{bmatrix},\ 
        C_4=\begin{bmatrix}
        1& 1& 0& 0& 0& 0
        \end{bmatrix}.
    \end{aligned}
\end{equation*}
It is straightforward to verify that $\mathrm{rank}(C_i\bar{B}_i)\neq\mathrm{rank}(\bar{B}_i)$, $\forall i\in\{1,\cdots,4\}$. Therefore, this problem cannot be solved using existing approaches \cite{yang2022state,cao2023distributed,cao2023distributedauto,zhu2024distributed,disaro2024distributed}.
\begin{figure}[ht]
    \centering
    \scalebox{0.7}{
    \begin{tikzpicture}
\def\off{18}
\def\N{5}
\def\R{2.5}
\pgfmathparse{360/\N}
\edef\step{\pgfmathresult}

\colorlet{net}{teal}
\tikzset{comm/.style = {color=net, very thick, dash pattern=on 4pt off 1.5pt}}
\node [circle, 
        draw, 
        color=net, 
        fill=white, 
        text=black, 
        very thick,
        inner sep=6pt] (O1) at (0,2) {$\mathcal O_{1}$};
\node [circle, 
        draw, 
        color=net, 
        fill=white, 
        text=black, 
        very thick,
        inner sep=6pt] (O2) at (2,2) {$\mathcal O_{2}$};
\node [circle, 
        draw, 
        color=net, 
        fill=white, 
        text=black, 
        very thick,
        inner sep=6pt] (O3) at (2,0) {$\mathcal O_{3}$};
\node [circle, 
        draw, 
        color=net, 
        fill=white, 
        text=black, 
        very thick,
        inner sep=6pt] (O4) at (0,0) {$\mathcal O_{4}$};

\draw [comm] (O1) -- (O2) -- (O3) -- (O4) -- (O1);

\end{tikzpicture}}\\[-1.5ex]
    \caption{Communication topology for numerical example in Section~\ref{sec:numerical_sim}.}
    \label{fig:topo_for_num_exam}
\end{figure}
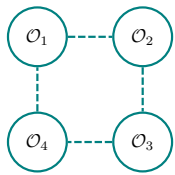
Each local observer $\mathcal{O}_i$ has access only to its local input $u_i$ and output $y_i$, and shares its state estimates with neighboring nodes over the communication network shown in Fig.~\ref{fig:topo_for_num_exam}.
According to Theorem~\ref{thm:main_ct}, $\chi_i$ is set to $279.0354$ and $\gamma_i$ is set to $30.7682$. The output injection maps $L_i$ and insertion maps $W_{g,i}^*$ are provided in the supplementary document\footnote{\label{docm:parameter}\href{https://github.com/RuixuanZhaoEEEUCL/Simulation-Parameters-DUIO.git}{https://github.com/RuixuanZhaoEEEUCL/Simulation-Parameters-DUIO.git}}.

From Fig.~\ref{fig:sim_result_CT}, we can see that the estimation errors of each node asymptotically converge to zero, verifying the effectiveness of our continuous-time DUIO design.
\begin{figure}[ht]
    \centering
    \includegraphics[width=0.49\textwidth]{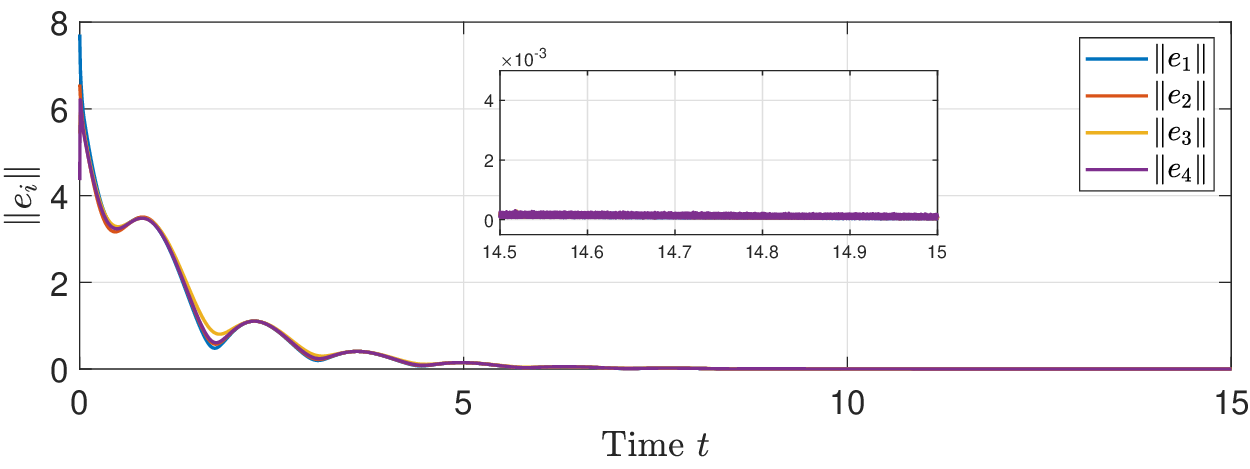}\\[-2ex]
    \caption{Norm of estimation error signal at each node for the continuous-time system.}
    \label{fig:sim_result_CT}
\end{figure}
It is noteworthy that stable invariant zeros may exist that are not captured in $\kappa \pr{A_{L_i}|\X/\Ss^*}$. For instance, at node 1, the spectrum of invariant zeros is $\kappa(A_{L_1}|\Ss_1^*/\W_1^*)=\{-2,\ 3.4641i,\ - 3.4641i\}$, which includes the stable zero $\{-2\}$. The state-space information associated with this invariant zero would be neglected if only the infimal unobservability subspace $\Ss_1^*$ is considered. Moreover, for the system in this example, $\overset{4}{\underset{i=1}{\bigcap}} \Ss_i^*\neq0$, which cannot guarantee a positive definite ${S^*}^\top(\mathcal{L} \otimes I_n)S^*$, where $S^*=\diag(S^*_1,\cdots,S^*_4)$, and $S_i^*:\Ss_i^*\rightarrow\X$ is the insertion map associated with $\Ss_i^*$. This highlights the advantage of designing via $\W_{g,i}^*$ rather than using $\Ss^*_i$ at each node, as the latter approach could result in design infeasibility.
Due to the presence of the $\text{sign}(\cdot)$ function, chattering could arise when the  observer~\eqref{eq:DUIO_ct} is employed. To mitigate chattering, several well-established smooth sign functions, such as those proposed in \cite{edwards1998sliding}, can be employed.

\subsubsection{Discrete-time system}
Similarly to the continuous-time case, we now consider a discrete-time system in the form of \eqref{eq:dis_sys}, where the system matrix $A$ and the input matrix $B$ are given by
\begin{equation*}
\setlength{\arraycolsep}{2.5pt}
\begin{aligned}
    &\ A =\\
    &\begin{bmatrix}
    0.9925&    0.1496&    0.0037&    0.0001&         0&         0&\\
   -0.0998&    0.9925&    0.0498&    0.0024&         0&         0\\
         0&         0&    0.9928&    0.0949&         0&         0\\
         0&         0&   -0.1424&    0.8978&         0&         0\\
    0.0002&   -0.0056&   -0.0037&    0.0472&    0.9850&   -0.1493\\
   -0.0037&    0.0744&    0.0016&    0.0049&    0.1990&    0.9850
    \end{bmatrix},
\end{aligned}
\end{equation*}
\begin{equation*}
\setlength{\arraycolsep}{1.5pt}
\begin{aligned}
&\ B^\top =
\begin{bmatrix}
    B_a&B_b&B_c
\end{bmatrix}^\top
=\\
&\begin{bmatrix}
    0.0037&    0.0499&         0&         0&    0.0497&    0.0069\\
         0&         0&    0.0024&    0.0475&    0.0012&    0.0001\\
    0.0001&    0.0012&    0.0499&   -0.0036&   -0.0038&    0.0498
\end{bmatrix}.
\end{aligned}
\end{equation*}
The system is driven by the input signal \( u(t) = \begin{bmatrix} \mathtt{sin}(0.01t) & \mathtt{cos}(0.05t) & 0.5\mathtt{sin}(0.05t) \end{bmatrix}^\top \), which can be decomposed as \( u = \begin{bmatrix} u_a & u_b & u_c \end{bmatrix}^\top \). The locally known inputs for each node are specified as \( u_1 = u_b \), \( u_2 = \begin{bmatrix} u_a & u_c \end{bmatrix}^\top \), \( u_3 = u_c \), and \( u_4 = \begin{bmatrix} u_b & u_c \end{bmatrix}^\top \). Based on this, the known and unknown input channels for each node are given below
\begin{equation*}
    \begin{aligned}
        &B_1 = B_b = \bar{B}_2,\ \bar{B}_1=\matrices{B_a&B_c} = B_2,\\
        &B_3 = B_c,\ \bar{B}_2 = \matrices{B_a&B_b},\ 
        B_4 = \matrices{B_b&B_c},\ \bar B_4=B_a.
    \end{aligned}
\end{equation*}
The global system is measured by $4$ networked sensors with the output matrices of each node 
\begin{equation*}
    \begin{aligned}
        &C_1=\begin{bmatrix}
        1& 0& 0& 0& 0& 0
        \end{bmatrix},\ 
        C_2=\begin{bmatrix}
        0& 1& 0& 0& 0& 0
        \end{bmatrix},\\
        &C_3=\begin{bmatrix}
        0& 0& 1& 0& 0& 0
        \end{bmatrix},\ 
        C_4=\begin{bmatrix}
        0& 0& 0& 1& 0& 1
        \end{bmatrix}.
    \end{aligned}
\end{equation*}
The sensor nodes are interconnected according to the graph shown in Fig.~\ref{fig:topo_for_num_exam}. The DUIO is designed by setting the communication frequency $d$ to $12$. In the meantime, the output injection maps $L_i$, canonical projections $P_{W_{g,i}^*}$ and matrices $E_i$, $F_i$ are given in the supplementary document\textsuperscript{\ref{docm:parameter}}.

The estimation errors of the DUIO are shown in Figure~\ref{fig:sim_result_DT}. As observed, all four error signals asymptotically converge with negligible steady-state error, consistent with Theorem~\ref{theorem:dis_DUIO}.
\begin{figure}[ht]
    \centering
    \includegraphics[width=0.49\textwidth]{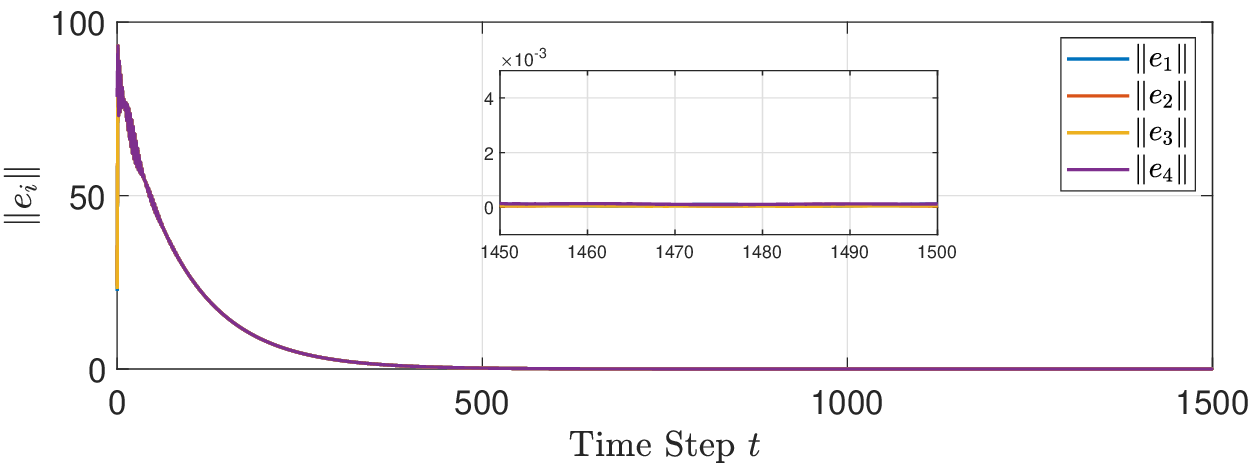}\\[-2ex]
    \caption{Norm of estimation error signal at each node for the discrete-time system.}
    \label{fig:sim_result_DT}
\end{figure}
Moreover, at node $2$, we have $\kappa(A_{L_2}|\Ss_2^*/\W_2^*)=\{0.9850 + 0.1724i,\ 0.9850 - 0.1724i,\ -0.0152,\ 0.9999\}$. Although the invariant zero $\{0.9999\}$ is stable, it is classified as a ``bad'' eigenvalue and is included in $\bar \X^*_{b,2}$ due to its slow convergence rate. This indicates the flexibility of the criterion to distinguish $\bar \X^*_{b,i}$ and $\bar \X^*_{g,i}$, which is crucial for the design of $\W_{g,i}^*$.

\subsection{Case Study: DC Microgrid}\label{exam:DGU}

Consider a DC microgrid power system composed of five interconnected DGUs, where buck converters are physically connected through RLC filters. The diagram of the DC microgrid power system with DUIO $\{ \mathcal O_i \}_{i \in \mathbf \{1,\cdots,5\}}$ is shown in Fig.~\ref{fig:MG}.
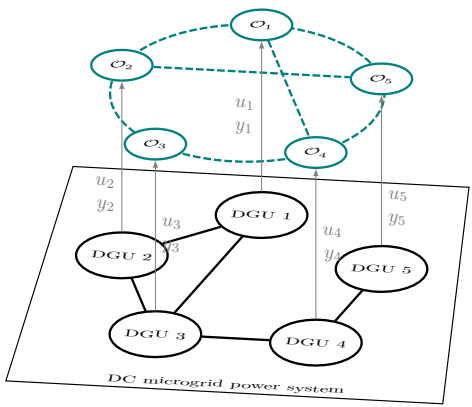
\begin{figure}[ht]
    \centering
    \scalebox{0.73}{\begin{tikzpicture}[tdplot_main_coords]
    \def\off{18}
    \def\N{5}
    \def\R{2.5}

    \pgfmathparse{360/\N}
    \edef\step{\pgfmathresult}

    \colorlet{net}{teal}
    \tikzset{comm/.style = {color=net, very thick, dash pattern=on 4pt off 1.5pt}}
    \foreach \x in {1,...,\N} {%
        \pgfmathparse{\R*cos(\x*\step+\off)}
        \let\xcoord\pgfmathresult
        \pgfmathparse{\R*sin(\x*\step+\off)}
        \let\ycoord\pgfmathresult
        \begin{scope}[canvas is xy plane at z=0,transform shape]
        \node [circle, 
            draw,
            color=black, 
            fill=white, 
            text=black, 
            very thick,
            inner sep=6pt] 
            (N\x) at (\xcoord, \ycoord, 0) {$\mathrm{DGU}\ \x$};
        \end{scope}
    }
    \draw[-latex, very thick, -] (N1) to (N2);
    \draw[-latex, very thick, -] (N2) to (N3);
    \draw[-latex, very thick, -] (N3) to (N4);
    \draw[-latex, very thick, -] (N1) to (N3);
    \draw[-latex, very thick, -] (N4) to (N5);
    \draw[semithick] (-4,-4,0) -- (4,-4,0) -- (3.65,3.9,0) -- (-3.6,3.9,0) -- cycle;
    \begin{scope}[canvas is xy plane at z=0,transform shape]
        \node[above] at (0,-4,0) {DC microgrid power system};
    \end{scope}

    \foreach \x in {1,...,\N} {%
        \pgfmathparse{\R*cos(\x*\step+\off)}
        \let\xcoord\pgfmathresult
        \pgfmathparse{\R*sin(\x*\step+\off)}
        \let\ycoord\pgfmathresult
        \begin{scope}[canvas is xy plane at z=4,transform shape]
        \node [circle, 
            draw,
            color=net, 
            fill=white, 
            text=black, 
            very thick,
            inner sep=6pt] 
            (O\x) at (\xcoord, \ycoord, 4) {$\mathcal{O}_\x$};
        \end{scope}
    }
    \coordinate (D) at (0,0,4);
\pgfmathanglebetweenpoints{
    \pgfpointanchor{O1}{west}}{
    \pgfpointanchor{D}{center}}
\edef\anglestart{\pgfmathresult}
\pgfmathanglebetweenpoints{
    \pgfpointanchor{O2}{north}}{
    \pgfpointanchor{D}{center}}
\edef\angleend{\pgfmathresult-11}
\begin{scope}[canvas is xy plane at z=4,transform shape]
    \foreach \x in {1,...,\N} {
    \edef\rotation{\step*(\x-1)}
    \draw[comm,rotate=\rotation] (\anglestart-180:\R) arc[radius = \R, start angle = \anglestart-180, end angle = \angleend-180];
    }
\end{scope}

\draw [comm] (O2) -- (O5);
\draw [comm] (O1) -- (O4);

\draw [gray,-latex] (N1) -- (O1) node[midway, left, align=left] {$u_1$ \\ $y_1$};
\draw [gray,-latex] (N2) -- (O2) node[near start, left, align=right] {$u_2$ \\ $y_2$};
\draw [gray,-latex] (N3) -- (O3) node[midway, right, align=left] {$u_3$ \\ $y_3$};
\draw [gray,-latex] (N4) -- (O4) node[midway, right, align=right] {$u_4$ \\ $y_4$};
\draw [gray,-latex] (N5) -- (O5) node[near start, right, align=left] {$u_5$ \\ $y_5$};

\end{tikzpicture}}\\[-1.5ex]
    \caption{Diagram of a DC microgrid power system composed of 5 DGUs with the DUIO $\{ \mathcal O_i \}_{i \in \mathbf \{1,\cdots,5\}}$. Solid lines represent the physical connections, while dashed lines represent the communication graph.}
    \label{fig:MG}
\end{figure}
The dynamics of an individual DGU are modeled by~\cite{tucci2017line}:
\begin{equation*}
\left\{\begin{aligned}
&\dot{\mathcal{V}}_{i}=\frac{1}{C_{t i}} \mathcal{I}_{t i}+\sum_{j \in \mathcal{N}_i}\left(\frac{\mathcal{V}_{j}}{C_{t i} R_{i j}}-\frac{\mathcal{V}_{i}}{C_{t i} R_{i j}}\right)-\frac{1}{C_{t i}} \mathcal{I}_{L i} \, , \\
&\dot{\mathcal{I}}_{t i}=-\frac{1}{L_{t i}} \mathcal{V}_{i}-\frac{R_{t i}}{L_{t i}} \mathcal{I}_{t i}+\frac{1}{L_{t i}} \mathcal{V}_{t i} \, ,
\end{aligned}\right.\label{eq:DGU_c}
\end{equation*}
where $\mathcal{V}_{ti}$ is the terminal voltage, $\mathcal{I}_{Li}$ is the load current, $\mathcal{V}_i$ is the $i$-th point of common connection voltage, $\mathcal{I}_{ti}$ is the filter current, $R_{ti}$, $L_{ti}$, $C_{ti}$ are the capacitance, inductance, and resistance of the DGU's RLC filter, respectively. $R_{ij}$ is the resistance of power line $ij$, and $\mathcal{N}_i$ is the set of DGUs physically connected with DGU $i$. To ensure that $\mathcal{V}_i$ can asymptotically track a reference voltage, each DGU is equipped with a primary controller proposed in \cite{tucci2017line}. In this context, the electrical scheme of each DGU is shown in Fig.~\ref{fig:DGU_circuit}.
\begin{figure}[ht]
    \centering
    \scalebox{0.65}{\begin{circuitikz}[american]
\draw
    (0,2) to[battery, o-o] (0,0) 
    (0,0) to[short] (0.75,0)
    (0,2) to[short] (0.75,2) 
    (1.5,1) node[draw, minimum width=1.5cm, minimum height=3cm] (buck) {Buck $i$}
  
    (2.25,2) to[R=$R_{ti}$, i^=$\mathcal{I}_{ti}$] (5,2)
    
    (5,2) to[L=$L_{ti}$] (6.5,2)

    (6.5,2) to[C=$C_{ti}$] (6.5,0)

    (8.5,2) to[isource, i^=$\mathcal{I}_{Li}$] (8.5,0)
    (2.25,0) to[short] (8.5,0)
    (6.5,2) to[short,-o] (7.5,2)
    (7.5,2) to[short,o-] (8.5,2)
    (7.5,2) node[above] {$\mathcal{V}_i$}

    (8.5,2) to[R=$R_{ij}$,-o] (11,2)
    (8.5,0) to[short,-o] (11,0)
    (11,2) node[above] {$\mathcal{V}_j$}
    
    [->] (2.5,0.25) -- (2.5,1.75) node[midway, right]{$\mathcal{V}_{ti}$}
   ;

    \draw[thick] (5,-2.5) -- (6,-2.0) -- (6,-3.0) -- cycle;
    \node[right] at (5.3,-2.5) {$K_i$};

    \draw[->, thick, green] (5,2) -- (5,4.3) -- (12,4.3) -- (12,-2.8) -- (6,-2.8);
    \node[circle, 
        draw, 
        color=green, 
        very thick,
        inner sep=1pt] at (5,2){};

    \draw[->, thick, red] (8,2) -- (8,4) -- (11.5,4) -- (11.5,-2.5) -- (6,-2.5);
    \node[circle, 
        draw, 
        color=red, 
        very thick,
        inner sep=1pt] at (8,2){};

    \draw[->, thick, red] (11.5,-2.2) -- (9.98,-2.2);
    \node[circle, 
        draw, 
        color=black, thick,
        inner sep=5.5pt] at (9.7,-2.2){};
    \node[right, font=\small] at (9.6,-2.2) {$-$};
    \node[above, font=\small] at (9.7,-2.27) {$+$};
    \draw[->, thick, black] (9.44,-2.2) -- (8,-2.2);
    \node[above] at (9.25,-2.3) {$\dot{\upsilon}_i$};

    \draw[->, thick, black] (9.7,-1.5) --  (9.7,-1.92);
    \node[right] at (9.7,-1.5) {$\mathcal{V}_{ref,i}$};
    
    \node[rectangle, 
        draw, 
        color=black, 
        very thick,
        inner sep=1pt] at (7.85,-2.2){$\frac{1}{s}$};

    \draw[->, thick, black] (7.7,-2.2) --  (6,-2.2);

     \draw[->, thick, blue] (5,-2.5) --  (4,-2.5) -- (4,1) -- (3.2,1)   ;
    
    \draw[dashed] (-0.5,-1) rectangle (9.1,3.5);
    \draw[dashed] (9.1,-1) rectangle (11.25,3.5);

    \node at (4.45,3.7) {DGU $i$};
    \node at (10.2,3.7) {Power Line $ij$};

\end{circuitikz}}\\[-1.5ex]
    \caption{Circuit diagram of DGU $i$ with connecting line and primary controller \cite{tucci2017line}, where $\dot{\upsilon}_i=V_{ref,i}-V_i$, the red line transmits $V_i$, the green line transmits $I_{ti}$ and the blue line transmits $K_i\matrices{V_i&I_{ti}&\upsilon_i}^\top$ with $K_i = \matrices{k_{i,1}&k_{i,2}&k_{i,3}}$.}
    \label{fig:DGU_circuit}
\end{figure}
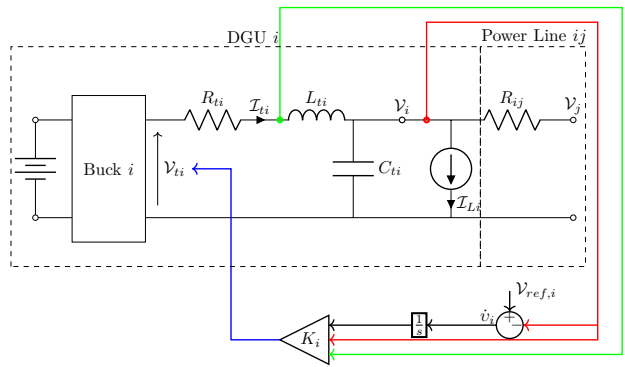

To account for the shared load among parallelly operated DGUs \cite{khadem2011parallel,dhua2017study}, we assume DGU $1$-$3$ supply a common load, i.e., $\mathcal{I}_{L1}=\mathcal{I}_{L2}=\mathcal{I}_{L3}=\mathcal{I}_1$, and track the same reference voltage, i.e., $\mathcal{V}_{ref,1}=\mathcal{V}_{ref,2}=\mathcal{V}_{ref,3}=\mathcal{V}_{r1}$, while DGU $4$ and $5$ supply another load, i.e., $\mathcal{I}_{L4}=\mathcal{I}_{L5}=\mathcal{I}_2$ and track a reference voltage, i.e., $\mathcal{V}_{ref,4}=\mathcal{V}_{ref,5}=\mathcal{V}_{r2}$. Each DGU has access only to its local input $u_i=\matrices{\mathcal{I}_{L_i}&\mathcal{V}_{ref,i}}^\top$ and local measurement $y_i=\matrices{\mathcal{V}_i&\mathcal{I}_{ti}}^\top$, as illustrated in Fig.~\ref{fig:MG}. Let $x=\underset{i\in\{1,\cdots,5\}}{\col}\left(\matrices{\mathcal{V}_i&\mathcal{I}_{ti}&\upsilon_i}^\top\right)$,  $u=\matrices{\mathcal{I}_1&\mathcal{V}_{r1}&\mathcal{I}_2&\mathcal{V}_{r2}}^\top$ and $y=\col_{i\in\{1,\cdots,5\}}\left(\matrices{\mathcal{V}_i&\mathcal{I}_{ti}}^\top\right)$. The overall system dynamics then follow \eqref{eq:con_sys} with 
$A = [A_{ij}]_{\,i,\,j\in \{1,\ldots,5\}}\in \mathbb{R}^{15 \times 15}$ and $B = [B_{ij}]_{\,i\in \{1,\ldots,5\},\,j\in \{1,2\}} \in \mathbb{R}^{15 \times 4}$, where 
\begin{equation*}
    A_{ii} = \begin{bmatrix}
        -\sum_{i\in\mathcal{N}_i}\frac{1}{R_{ij}C_{ti}}&\frac{1}{C_{ti}}&0\\
        \frac{k_{i,1}-1}{L_{ti}}&\frac{k_{i,2}-R_{ti}}{L_{ti}}&\frac{k_{i,3}}{L_{ti}}\\
        -1&0&0
    \end{bmatrix},
\end{equation*}
\begin{equation*}
    A_{ij}=\left\{\begin{array}{lr}
        \mathbf{0}_{3\times 3},\ \mathrm{if}\ j\neq i\ \&\ i\notin\mathcal{N}_i\\
        \begin{bmatrix}
            \frac{1}{R_{ij}C_{ti}}&0&0\\
            0&0&0\\
            0&0&0
        \end{bmatrix},\ \mathrm{if}\ j\neq i\ \&\ i\in\mathcal{N}_i
    \end{array}\right.\,.
\end{equation*}
and
\begin{equation*}
    B_{ij}=\left\{\begin{array}{ll}
        \mathbf{0}_{3\times 2},\ & \mathrm{if}\ (j=1 \,\& \,i=\{4,5\}) \parallel 
 \\ & \hspace{9mm} (\ j=2 \,\&\, i=\{1,2,3\})\\
        \begin{bmatrix}
             -\frac{1}{C_{ti}}&0\\
            0&0\\
            0&1
        \end{bmatrix}\!\!,\ & \hspace{2.8cm} \mathrm{otherwise}
    \end{array}\right.
\end{equation*}
Local output matrix is
$
    C_i = [
        \mathbf{0}_{2\times 3(i-1)}\,\,I_2 \,\, \mathbf{0}_{2\times 1}\,\,\mathbf{0}_{2\times 3(5-i)}].
$
In view of \eqref{eq:inputdecomp}, the input matrices for known and unknown inputs at each sensor node are 
\begin{equation*}
\begin{aligned}
    &B_1=B_2=B_3=\bar{B}_4=\bar{B}_5 =\underset{i\in\{1,\cdots,5\}}{\col}(B_{i1}),\\
    &B_4=B_5=\bar{B}_1=\bar{B}_2=\bar{B}_3=\underset{i\in\{1,\cdots,5\}}{\col}(B_{i2}),
\end{aligned}
\end{equation*}
In this example, the rank condition \eqref{eq:rank_con} is also violated, since $\mathrm{rank}(C_i\bar{B}_i)=0\neq\mathrm{rank}(\bar{B}_i)$. As a result, many existing DUIO schemes are not applicable. 

Given the parameters of this example as listed in TABLE~\ref{tab:mG_para}.
\begin{table}[ht]
    \centering
    \caption{Electrical Parameters of DGU}
  \resizebox{0.45\textwidth}{!}{
    \begin{tabular}{c|c|c}
                        \hline
        Parameter & Symbol  & Value \\
                         \hline
        Output capacitance&$C_{ti}$&$2.2\ \mathrm{mF}$\\
        Converter inductance &$L_{ti}$&$1.8\ \mathrm{mF}$\\
        Loss resistance   &$R_{ti}$&$0.2\ \Omega$\\
        Power line resistance & $R_{ij}$ & $0.05\ \Omega$\\
        Primary Controller&$K_i$&$[-2.134\ -0.163\ 13.553]$\\
        \hline
    \end{tabular}
    }
    \label{tab:mG_para}
\end{table}
Here, we focus on testing the discrete-time observer, as its design conditions are generally more relaxed compared to the continuous-time counterpart, as discussed in Section~\ref{sec:discrete_DUIO}.
To do this, we first derive the discrete-time model of the five interconnected DGU systems using exact discretization: 
$
    A_\mathrm{d} = e^{A T_s},\ B_\mathrm{d} = \left(\int_{\tau=0}^{T_s}e^{A\tau}\right)B,\ C_\mathrm{d} = C
$
with a sampling time $T_s=1\,\mathrm{ms}$. Next, the local discrete-time observer is designed according to \eqref{eq:dis_DUIO}. Due to space limitations, all DUIO parameters are provided in the supplementary document\textsuperscript{\ref{docm:parameter}}.


The results are demonstrated in Fig.~\ref{fig:d10_noise_free}. As observed, with $d=16$, the state estimation error of each observer converges with only a negligible steady-state error. To further assess the robustness, we introduce Gaussian sensor noise $\mathsf{N}(0,0.1^2)$ into the local measurements and perturb the load inputs $\mathcal{I}_1$ and $\mathcal{I}_2$ with Gaussian noise corresponding to  $\mathrm{SNR}=30\ \mathrm{dB}$, representing forecast inaccuracies. Fig.\ref{fig:d10_noised} demonstrates the robust performance of the proposed DUIO scheme under both sensor and input noise.  
\begin{figure}[ht]
    \centering
    \subfloat[Noise-free case with $d=16$]{
    \includegraphics[width=0.49\textwidth]{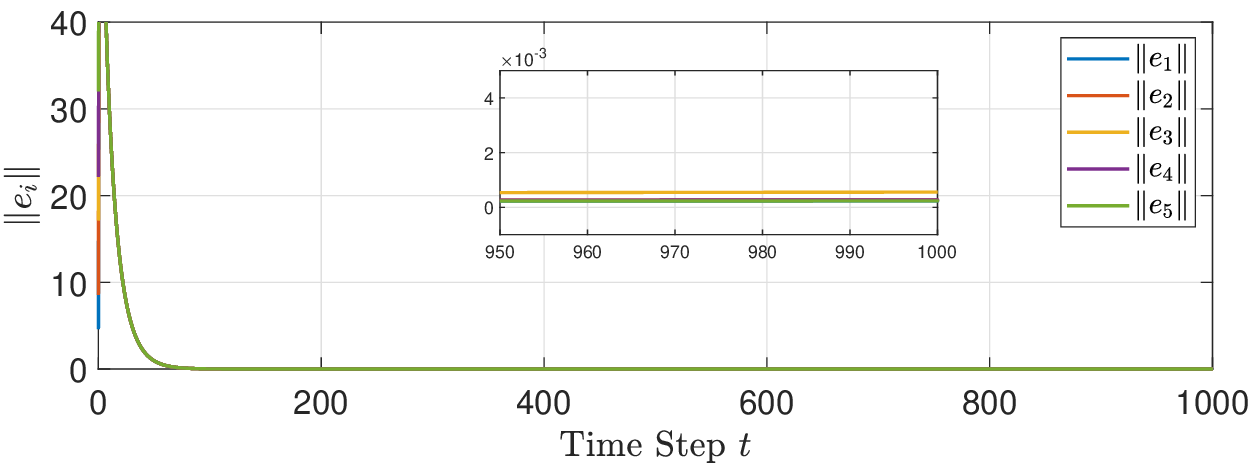}
    \label{fig:d10_noise_free}
    }
    \hfill \\[-0.5ex]
    \subfloat[Noised case with $d=16$]{
    \includegraphics[width=0.49\textwidth]{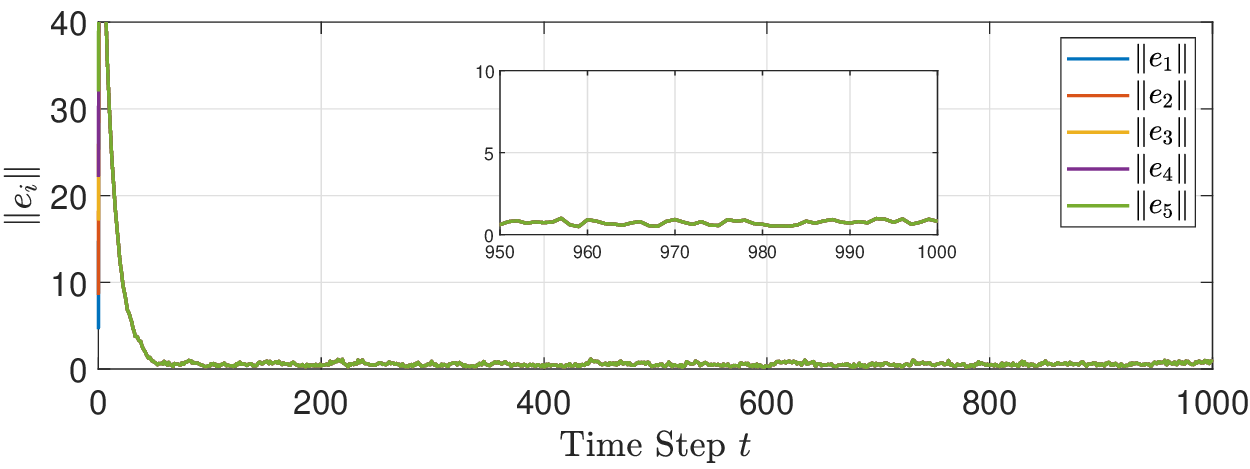}
    \label{fig:d10_noised}
    }
    \\[-1.2ex]
    \caption{Norm of DUIO estimation error signal for the DC microgrid system under different scenarios.}
    \label{fig:results_DGU_case}
\end{figure}
\begin{table}[ht]
    \centering
    \caption{Estimation Error under different communication frequency of DC microgrid power system}
  \resizebox{0.25\textwidth}{!}{
    \begin{tabular}{c|c|c}
                        \hline
        $d$ & $\|x(t)\|$  & $\underset{t\rightarrow\infty}{\lim} \|e(t)\|$ \\
                         \hline
        $10$ &$357.3444$&$0.2685$\\
        $12$&$357.3444$&$0.0493$\\
        $14$&$357.3444$&$0.0091$\\        
        $16$&$357.3444$&$0.0017$\\
        \hline
    \end{tabular}
    }  
    \label{tab:DGU_frequency_compare}
\end{table}

Finally, the influence of communication frequency is investigated in TABLE~\ref{tab:DGU_frequency_compare}. In line with Theorem~\ref{theorem:dis_DUIO}, estimation accuracy improves as $d$ increases, though at the expense of greater communication resource usage. In practice, this trade-off must be carefully managed to balance accuracy requirements against available system resources.

\section{Concluding Remarks}\label{sec:conclusion}

In this paper, we present a novel Distributed Unknown Input Observer (DUIO) design for both continuous-time and discrete-time linear systems. Using a geometric approach, we address the challenge of distributed state estimation under unknown inputs across multiple nodes in a networked system. A key feature of our approach is the introduction of joint geometric conditions, which enable the estimation of a reduced-order linear projection of the system state. These projections allow each sensor node to extract the most informative components unaffected by locally unknown inputs, thereby relaxing the restrictive rank conditions imposed by existing approaches while preserving estimation accuracy. The framework is further extended to discrete-time systems—constituting the first discrete-time DUIO scheme—where a two-time-scale structure is proposed to separate local estimation and inter-node communication. Our analysis shows that steady-state accuracy improves with a higher communication frequency. The effectiveness and practical relevance of the proposed methods are demonstrated through extensive simulations, including a power grid case study. 

Future research will focus on extending our methodology to nonlinear systems, reconfigurable sensor networks, disturbance observers, and extended state observers.


\begin{IEEEbiography}[{\includegraphics[width=1in,height=1.25in,clip,keepaspectratio]{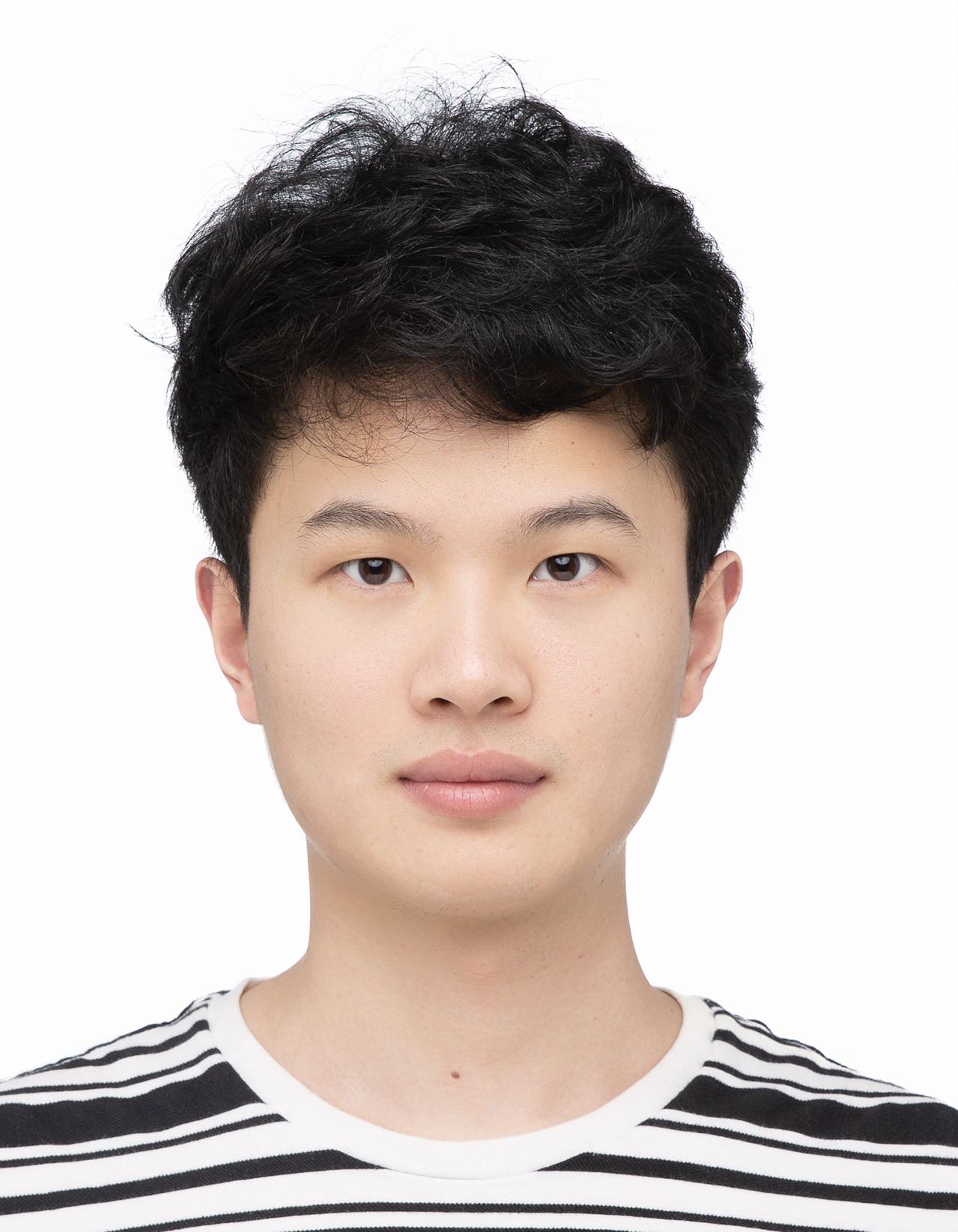}}]{Ruixuan Zhao}
obtained the M.Sc. degree in Control and Optimisation from Imperial College London, UK, in 2022, and the B.Eng degree in Automation Engineering from Nanjing University of Aeronautics and Astronautics, China, in 2020. He is currently pursuing a Ph.D. in the Department of Electronic and Electrical Engineering at the University College London (UCL), UK.
His research interests focus on state estimation, distributed observer design, geometric approach, and the sensor networks for cyber-physical systems.
\end{IEEEbiography}

\begin{IEEEbiography}[{\includegraphics[width=1in,height=1.25in,clip,keepaspectratio]{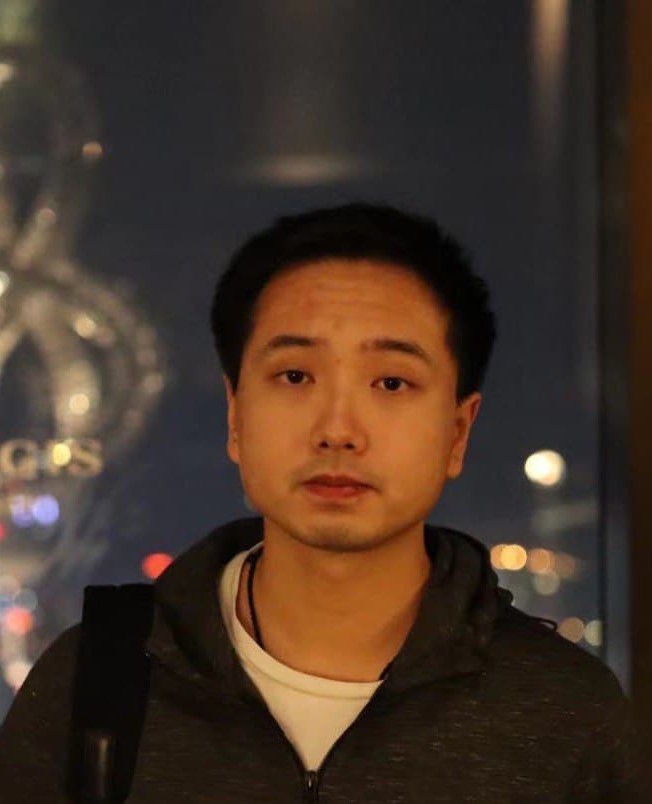}}]{Guitao Yang}
received the B.Eng. degree in Electrical and Electronic Engineering from The University of Manchester, UK, in 2016. He received the M.Sc. and Ph.D. degrees in Control Systems from Imperial College London, UK, in 2017 and 2022, respectively. From 2022 to 2025, he was a Research Associate with the Department of Electrical and Electronic Engineering, Imperial College London. He is currently a Research Associate with the Wolfson School of Mechanical, Electrical and Manufacturing Engineering, Loughborough University, UK. His research interests include distributed state estimation, fault-tolerant observers, and the geometric approach.
\end{IEEEbiography}

\begin{IEEEbiography}[{\includegraphics[width=1in,height=1.25in,clip,keepaspectratio]{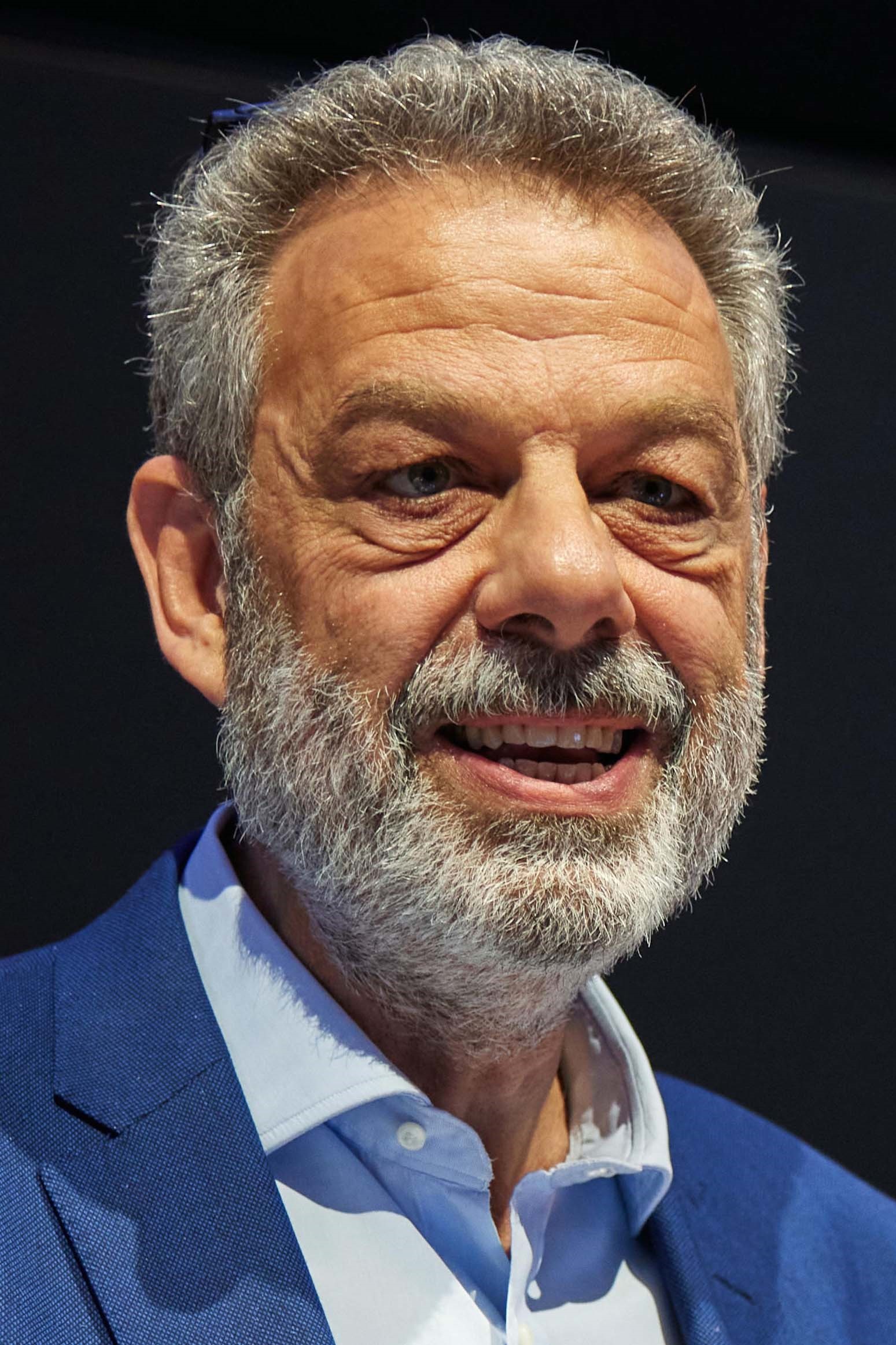}}]{Thomas Parisini} (Fellow, IEEE) received the Ph.D. degree in electronic engineering and computer science from the University of Genoa, Italy, in 1993. He was an Associate Professor with Politecnico di Milano, Milano, Italy. He currently holds the Chair of industrial control and is the Head of the Control and Power Research Group, Imperial College London, London, U.K. He also holds a Distinguished Professorship at Aalborg University, Denmark. Since 2001, he has been the Danieli Endowed Chair of automation engineering with the University of Trieste, Trieste, Italy, where from 2009 to 2012, he was the Deputy Rector. In 2023, he held a “Scholar-in-Residence”visiting position with Digital Futures-KTH, Stockholm, Sweden. He has authored or coauthored a research monograph in the Communication and Control Series, Springer Nature, and more than 400 research papers in archival journals, book chapters, and international conference proceedings. Dr. Parisini was the recipient of the Knighthood of the Order of Merit of the Italian Republic for scientific achievements abroad awarded by the Italian President of the Republic in 2023. In 2018 he received the Honorary Doctorate from the University of Aalborg, Denmark and in 2024, the IEEE CSS Transition to Practice Award. Moreover, he was awarded the 2007 IEEE Distinguished Member Award, and was co-recipient of the IFAC Best Application Paper Prize of the Journal of Process Control, Elsevier, for the three-year period 2011-2013 and of the 2004 Outstanding Paper Award of IEEE TRANSACTIONS ON NEURAL NETWORKS. In 2016, he was awarded as Principal Investigator with Imperial of the H2020 European Union flagship Teaming Project KIOS Research and Innovation Centre of Excellence led by the University of Cyprus with an overall budget of over 40 million Euros. He was the 2021-2022 President of the IEEE Control Systems Society and he was the Editor-in-Chief of IEEE TRANSACTIONS ON CONTROL SYSTEMS TECHNOLOGY (2009-2016). He was the Chair of the IEEE CSS Conference Editorial Board (2013-2019). Also, he was the associate editor of several journals including the IEEE TRANSACTIONS ON AUTOMATIC CONTROL and the IEEE TRANSACTIONS ON NEURAL NETWORKS. He is currently an Editor of Automatica and the Editor-in-Chief of the European Journal of Control. He was the Program Chair of the 2008 IEEE Conference on Decision and Control and General Co-Chair of the 2013 IEEE Conference on Decision and Control. He is a Fellow of IFAC. He serves as Chair of the IEEE CSS Awards and is a Member of IEEE TAB Periodicals Review and Advisory Committee.
\end{IEEEbiography}

\begin{IEEEbiography}[{\includegraphics[width=1in,height=1.25in,clip,keepaspectratio]{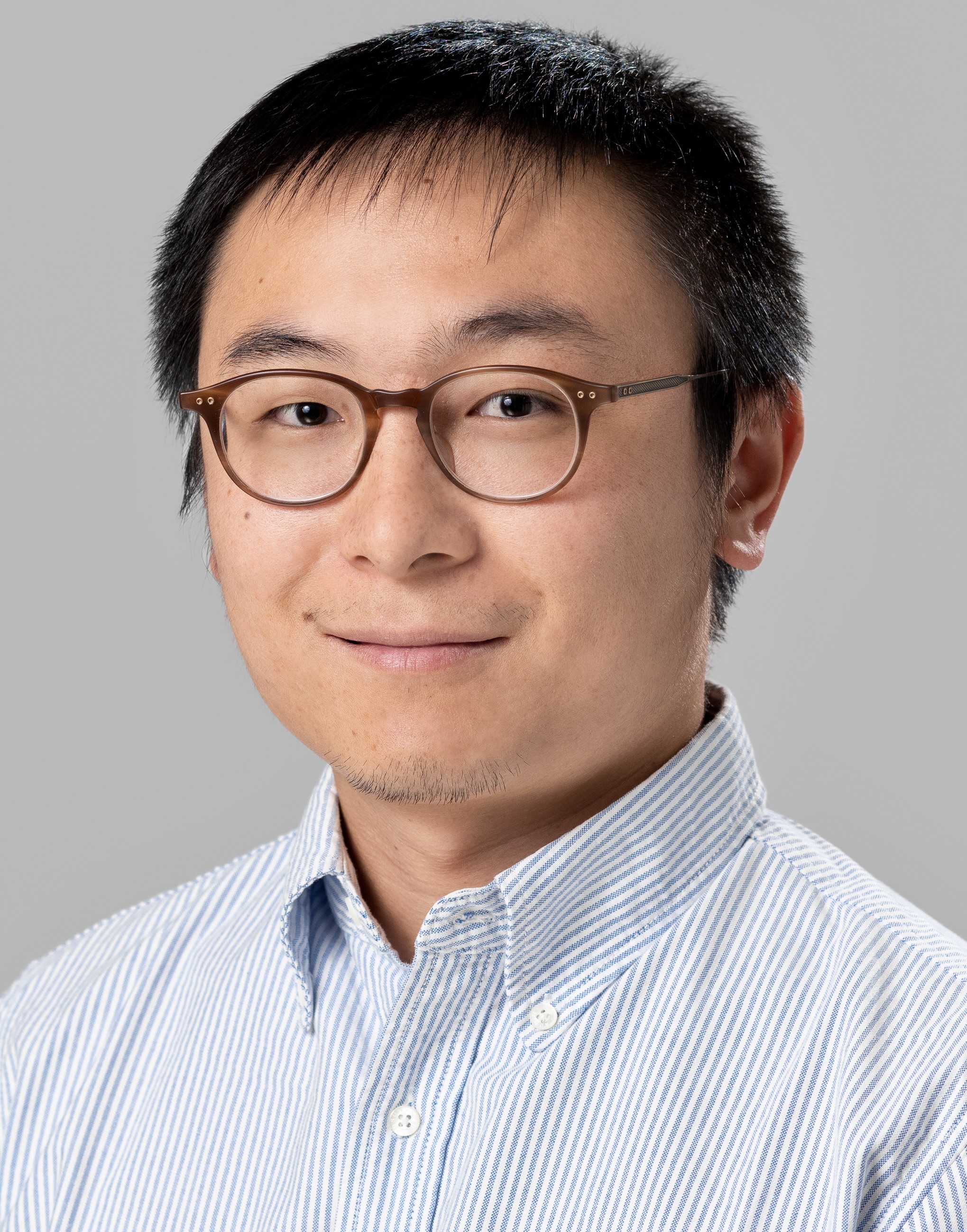}}]{Boli Chen}
(M'16 - SM'24) received his B.Eng. in Electrical and Electronic Engineering from Northumbria University, UK, in 2010. He earned his MSc and PhD in Control Systems from Imperial College London, UK, in 2011 and 2015, respectively. He is currently a Lecturer in the Department of Electronic and Electrical Engineering at University College London (UCL), UK. His research focuses on the control, optimization, and estimation of complex dynamical systems, with rich applications in smart cities, e.g., transportation, electric energy systems, and sensor networks. 
Dr Boli Chen is a member of the IEEE Control Systems Society Technical Committee on ``Smart Cities’’. He serves as an Associate Editor for the IEEE Transactions on Intelligent Transportation Systems and the European Journal of Control. Additionally, he is a member of the EUCA Conference Editorial Board and the IEEE ITSC Editorial Board.
\end{IEEEbiography}

\end{document}